\documentclass{article}

\usepackage{graphicx} 
\usepackage{adjustbox} 
\usepackage{enumerate} 
\usepackage{geometry} 
\usepackage{amsmath} 
\usepackage{amssymb} 
\usepackage[mathletters]{ucs} 
\usepackage[utf8x]{inputenc} 
\usepackage{fancyvrb} 
\usepackage{grffile} 
\usepackage{hyperref}
\usepackage{authblk}
\usepackage{booktabs}  

\usepackage{mathtools}
\usepackage{soul}
\usepackage{float}

\usepackage{listings}

\usepackage[format=plain, labelfont={it}, textfont=it]{caption}

\title{Explicit no arbitrage domain for sub-SVIs via reparametrization}

\author[1]{Claude Martini\thanks{cmartini@zeliade.com}}
\author[1, 2]{Arianna Mingone\thanks{arianna.mingone@polytechnique.edu}}

\affil[1]{Zeliade Systems, 56 rue Jean-Jacques Rousseau, Paris, France}
\affil[2]{Centre de Math\'ematiques Appliqu\'ees (CMAP), CNRS, Ecole Polytechnique, Institut Polytechnique de Paris}

\hypersetup{ pdftitle={Explicit no arbitrage domain for sub-SVIs via reparametrization} }

\usepackage{amsthm}
\usepackage{cleveref}
\newtheorem{theorem}{Theorem}[section]

\newtheorem{lemma}[theorem]{Lemma}
\newtheorem{proposition}[theorem]{Proposition}
\newtheorem{remark}[theorem]{Remark}

\begin{document}
    
    \maketitle

	\begin{abstract}

The no Butterfly arbitrage domain of Gatheral SVI 5-parameters formula for the volatility smile has been recently described. It requires in general a numerical minimization of 2 functions altogether with a few root finding procedures. We study here the case of some sub-SVIs (all with 3 parameters): the Symmetric SVI, the Vanishing Upward/Downward SVI, and SSVI, for which we provide an explicit domain, with no numerical procedure required. 

\end{abstract}

\section{Introduction}\label{introduction}

	Gatheral SVI parametrization of a volatility smile reads:
\[w(k) = a+ b \bigl( \rho(k-m)+ \sqrt{(k-m)^2+\sigma^2} \; \bigl)\]
where \(w(k)\) is the so-called \emph{total variance}: the squared
implied volatility times the remaining time to maturity at the
log-forward moneyness \(k\).

	SVI is known to fit very well a large set of market data, and recently
the challenging question of characterizing no Butterfly arbitrage in SVI
has been solved with a handy parametrization of the no arbitrage
domain, leading to an efficient implementation of a calibration
algorithm ensuring no Butterfly arbitrage (\(\!\)\cite{martini2020no}).
So why would one care about \emph{sub}-SVIs, meaning SVI with some
parameters frozen, or re-parametrizations of SVI with less than 5
parameters, like the volatility slices of SSVI surfaces of Gatheral and
Jacquier (\(\!\)\cite{gatheral2014arbitrage})?

	In fact, for several reasons. The 1st one is that SVI might be \emph{too
rich} in the sense that an excellent fit could also be achieved in most
cases by sub-SVIs, with the additional benefit to stabilize the
variation of the calibrated optimal parameter from one day to another,
or between different maturity slices. A good theoretical reason to
suspect this follows from considering SVI smiles with \(\rho=0\) and
\(m \neq0\): indeed the correspondence with stochastic volatility models
dictates that models with a zero correlation should yield symmetric
smiles, which implies \(m=0\); in this sense SVI smiles with \(\rho=0\)
and \(m \neq 0\) should correspond to smiles which are not associated to
stochastic volatility models, and also probably not very likely to be met in
real-life market data. In this direction, one could also note that the
result by Gatheral and Jacquier (\(\!\)\cite{gatheral2011convergence})
that the Long Term Heston smile goes to SVI shows in fact that it goes
to a sub-SVI; indeed the SVI parameters are given by
\begin{align*}
	a = \frac{\theta}2 (1-\rho^2), & & b=\frac{\theta\varphi}{2}, & & \rho, & & m= -\frac{\rho}{\varphi}, & & \sigma=\frac{\sqrt{1-\rho^2}}{\varphi},
\end{align*}
so that the long-term smile depends eventually only on the 3 parameters $(\theta, \varphi, \rho)$, with the constraint $\rho=0 \implies m=0$ enforced.

	The 2nd reason is that it is difficult to obtain no Calendar Spread
arbitrage conditions on two SVI smiles attached to two different
maturities, as discussed in \cite{gatheral2014arbitrage}. This has been
achieved for smiles corresponding to SSVI parameters, which are sub-SVI
ones with 3 parameters instead of 5, as shown in
\cite{hendriks2017extended}. So in order to obtain tractable no
arbitrage SVI \emph{surfaces}, it may be required in practice to
restrict the set of SVI parameters. Note that, in relation to the 1st
point above, given the poor ability of SSVI to fit especially on the
short term, the right balance between fitting ability and tractability
might lie in-between SVI and SSVI, in some sub-SVI with 4 parameters.

	Another hope is that there might be simplifications, due to the special
structure of the different sub-SVIs under study, both in the
\emph{Fukasawa domain} which characterizes a weak no Butterfly arbitrage
property and constitutes a key step in the full characterization of no
Butterfly arbitrage, and in the last stage of the algorithm in
\cite{martini2020no}, which requires 2 minimizations to get the lower
bound of the no arbitrage domain for the SVI parameter \(\sigma\). We
recall that the computation of the boundary of the \emph{Fukasawa
domain}, which corresponds to the weak \emph{necessary condition} of no
arbitrage obtained by\ldots{} Fukasawa
(\(\!\)\cite{fukasawa2012normalizing}), requires also some root-finding
algorithm in general, much quicker than the last stage minimization
algorithms for \(\sigma\) though.

	In this work we study in detail the following \emph{sub-SVIs}:

\begin{enumerate}
\def\labelenumi{\arabic{enumi}.}
\item
  The \emph{Vanishing SVI} where \(a=0\) and \(\rho=\pm1\), with its 2
  flavors \emph{Vanishing upward} (\(\rho=1\)), and \emph{Vanishing
  downward} (\(\rho=-1\)); the second family may correspond to real-life
  smiles and is given by
  \[w(k) = b\bigl(-(k-m) + \sqrt{ (k-m)^2 + \sigma^2}\; \bigr)\]
\item
  The \emph{Symmetric SVI} where \(\rho=m=0\), so that:
  \[w(k) = a+ b\sqrt{ k^2 + \sigma^2}\]
\item
  SSVI (slices) given by:
  \[w(k) = \frac{\theta}{2} \bigl( 1 + \rho \varphi k +\sqrt{(\varphi k + \rho)^2+1-\rho^2} \bigr)\]
\end{enumerate}

	For all those sub-SVIs, our program is clear: trying to find out an
\emph{explicit} parametrization for the Fukasawa domain and, when
possible, of the full no arbitrage domain. Besides brute-force calculus
(that we do use in many circumstances), we heavily use a
\emph{re-parametrization} paradigm, that we illustrate below in the
context of the bound on \(\sigma\) discussed above.

	This stage requires indeed to compute the minimum of a function
\(\tilde f(l; \gamma,b,\rho,\mu)\), on an interval which depends on the
parameters \((\gamma,\rho)\); we know that \(\tilde f\) goes to infinity
at the bounds of this interval. Our strategy is then to study the
\emph{critical points} of \(\tilde f\). It turns out that in full
generality, the equation characterizing those critical points \(\bar l\)
reads: \[p(\bar l;\gamma,\rho,\mu) = b^2 q(\bar l;\gamma,\rho,\mu).\] In
some circumstances, we can then \emph{use \(\bar l\)} as a parameter,
and obtain \(b^2\) as
\(\frac{p(\bar l;\gamma,\rho,\mu)}{q(\bar l;\gamma,\rho,\mu)}\). In
other circumstances, we use the fact that the equation \(p-b^2 q=0\) is,
in full generality, \emph{quadratic} in \(\mu\), and we use the same
trick to back up \(\mu\) once \(\bar l\) is promoted to the status of a
parameter.

Now this is only the easy part of the story, since the critical points
of \(\tilde f\) may correspond to other local minima than the absolute
one, or, even worse, to local maxima. The hard part is to show that the
chosen critical point \(\bar l\), in a given domain, corresponds indeed
to the global minimum of \(\tilde f\); one way to prove this is that
there is a \emph{unique} solution to the critical points equation above.
We manage to prove this unicity for the Symmetric and Vanishing case,
and resort to a \emph{numerical proof} for SSVI.

	This leads us to the following results:

\begin{itemize}
\item
  in \cref{vanishing-svi}, we obtain a fully explicit parametrization of
  the no arbitrage domain for the \emph{Vanishing Upward} and the
  \emph{Vanishing Downward} SVI
  (\Cref{PropVanishingUp,PropVanishingDown});
\item
  in \cref{extremal-decorrelated-svi}, we find a parametrization of the no arbitrage domain for the \emph{Extremal Decorrelated} SVI (\Cref{PropExtremal});
\item
  in \cref{symmetric-svi}, we get the no arbitrage domain for the \emph{Symmetric} SVI (\Cref{PropSymmetric});
\item
  in \cref{ssvi}, we derive a quasi-explicit parametrization (modulo a
  zero of a one-dimensional function to be computed numerically) of the
  no arbitrage domain for SSVI (\Cref{PropSSVI}). We also re-visit the
  Long Term Heston SVI appoximation and show it is in fact of SSVI type;
  we prove in \Cref{propLTH} that it is indeed free of Butterfly
  arbitrage as soon as \(T\) is larger than some fully explicit
  threshold, which completes Gatheral and Jacquier result.
\end{itemize}

	Up to now and to the best of our knowledge, the only known
\emph{volatility model} (meaning: a formula for the implied volatility)
with an explicit no arbitrage domain was the SSVI slice, with the
(restrictive) conditions obtained by Gatheral and Jacquier
(\(\!\)\cite{gatheral2014arbitrage}), and extended to the
characterization of the full no arbitrage domain in the decorrelated
case in \cite{guo2016generalized}. To this extent, the present work is a
big leap forward, since we obtain 3 new families: the Vanishing
Downward/Upward one, the Symmetric one, and the \emph{correlated}
SSVI, with explicit no arbitrage domains (a single-variable boundary
function has to be computed numerically for SSVI).

	All our sub-SVIs are \emph{3 parameters} SVI; as discussed above, 3
parameters may be too little to produce a good fit on market data, so we
would say that the interest here is essentially of academic nature: we
hope that those families can help in the investigation of the
theoretical properties of volatility smiles, or come as handy
illustrations of those properties. Note though that the Vanishing
Downward SVI could have a practical application to the case of
decreasing market smiles which are often encountered for not-too-short
maturities. SSVI is used in practice; our investigation yields a
parametrization of the SSVI slices satisfying the no Butterfly arbitrage
property which is much more effective than the generic one presented in
\cite{martini2020no}.

We warmly thank Stefano De Marco and Antoine Jacquier for their remarks and comments on a preliminary version of this work.

\section{Notations and
preliminaries}\label{notations-and-preliminaries}

\subsection{Necessary and sufficient no Butterfly arbitrage conditions
for
SVI}\label{necessary-and-sufficient-no-butterfly-arbitrage-conditions-for-svi}

	In the whole article, we will refer to the results in
\cite{martini2020no}, that we quickly summarize here, with some changes
in notations.

	The general form for SVI is
\[SVI(k) = a + b\bigl(\rho(k-m) + \sqrt{(k-m)^2+\sigma^2}\bigr)\] where
\(a,\ m\in\mathbb{R}\), \(b\geq 0\), \(\rho\in[-1,1]\), \(\sigma\geq0\).
If \(b=0\) and $a>0$, SVI reduces to the Black-Scholes case, which is free of
arbitrage, while if $a=b=0$, we recover the trivial case. In the following we will consider \(b>0\).

We set \(\gamma=\frac{a}{b\sigma}\) and \(\mu=\frac{m}{\sigma}\). Let us
redefine the quantity
\[N(l;\gamma,\rho) := \gamma + \rho l+\sqrt{l^2+1}\] such that
\(SVI(k) = b\sigma N\bigl(\frac{k-m}{\sigma}\bigr)\). The derivatives of
\(N\) are \begin{align*}
N'(l;\rho) &= \rho + \frac{l}{\sqrt{l^2+1}}, & N''(l) &= \frac{1}{(l^2+1)^{\frac{3}{2}}}
\end{align*} so that \(N\) has a unique critical point which is a point
of minimum equal to \(l^*=-\frac{\rho}{\sqrt{1-\rho^2}}\). Since \(N\)
is a rescaled total variance, it must be positive, so the constraint on
the new variables are \(\gamma\geq-\sqrt{1-\rho^2}\), \(b\geq 0\),
\(\rho\in[-1,1]\), \(\mu\in\mathbb{R}\), \(\sigma\geq0\). When \(b=0\),
we recover the Black-Scholes case, which is free of Butterfly arbitrage
under \(\sigma>0\). Then, when formulating theorems of non-arbitrage, we
consider only the more difficult cases \(b>0\).

We define \begin{equation}\label{eqhgg2}
\begin{aligned}
h(l;\gamma,\rho,\mu) &:= 1-N'(l;\rho)\frac{l+\mu}{2N(l;\gamma,\rho)},\\
g(l;\rho) &:= \frac{N'(l;\rho)}{4},\\
g_2(l;\gamma,\rho) &:= N''(l)-\frac{N'(l;\rho)^2}{2N(l;\gamma,\rho)},
\end{aligned}
\end{equation} and
\(G_1(l) := G_{1+}(l)G_{1-}(l) := (h(l)-bg(l))(h(l)+bg(l))\).

In section 4.1 of \cite{martini2020no}, it is shown that the requirement that an SVI is (Butterfly) arbitrage-free is equivalent
to the requirement that the function \[G_1 + \frac{1}{2\sigma}bg_2\] is
non-negative. Fukasawa proved in \cite{fukasawa2012normalizing} that the
condition \(G_{1+}>0\) and \(G_{1-}>0\) are also necessary. Theorem 5.10
of \cite{martini2020no}, that we rewrite here, characterizes these
conditions in the case of SVI. The statement requires the definition of
the functions: \begin{align*}
L_-(l;\gamma,b,\rho) &:= 2N(l;\gamma,\rho)\bigl(\frac{1}{N'(l;\rho)}+\frac{b}{4}\bigr)-l,\\
g_{-(b,\rho)} &:= \bigl(\rho\sqrt{l^2+1} + l\bigr)^2\bigl(\sqrt{l^2+1}\bigl(\frac{1}{2}+\frac{b \rho}{4}\bigr) + \frac{bl}{4}\bigr)-\bigl(\rho l + \sqrt{l^2+1}\bigr).
\end{align*} When \(\gamma+\sqrt{1-\rho^2}>0\) and \(|\rho|<1\),
preliminary propositions show that under the case \(b(1\pm\rho)<2\),
there exist a unique \(l_-(\gamma,b,\rho)< l^*\) and a unique
\(l_-(\gamma,b,-\rho)< l^*\) such that
\(g_{-(b,\rho)}(l_-(\gamma,b,\rho))=\gamma\) and
\(g_{-(b,-\rho)}(l_-(\gamma,b,-\rho))=\gamma\). In such way, the
quantity
\[\tilde F(b,\rho):= \inf\bigl\{\gamma|-L_-(l_-(\gamma,b,-\rho);\gamma,b,-\rho)>L_-(l_-(\gamma,b,\rho);\gamma,b,\rho)\bigr\}\land -\sqrt{1-\rho^2}\]
is well-defined and it is called the Fukasawa threshold. If instead
\(b(1-\rho)=2\) (or \(b(1+\rho)=2\)), there exists a unique
\(l_-(\gamma,b,-\rho)< l^*\) (resp. \(l_-(\gamma,b,\rho)\)) such that
\(g_{-(b,-\rho)}(l_-)=\gamma\) (resp. \(g_{-(b,\rho)}(l_-)=\gamma\)).
When the former case arises while the latter does not, the Fukasawa
threshold is defined as
\(\tilde F(b,\rho):= \inf\bigl\{\gamma|-L_-(l_-(\gamma,b,-\rho);\gamma,b,-\rho)>-\frac{b\gamma}2\bigr\}\land -\sqrt{1-\rho^2}\).
Vice versa, when it is the second case to be active while the first is
not, the quantity is defined as
\(\tilde F(b,\rho):= \inf\bigl\{\gamma|\frac{b\gamma}2>L_-(l_-(\gamma,b,\rho);\gamma,b,\rho)\bigr\}\land -\sqrt{1-\rho^2}\).
Finally, when both cases are valid, so \(b=2\) and \(\rho=0\), the
threshold is \(\tilde F(2,0):=0\). The aforementioned theorem is:

	\begin{theorem}[SVI parameters $(\gamma,b,\rho,\mu,\sigma)$ fulfilling Fukasawa necessary no arbitrage conditions]\label{TheoFuk}
Assume $\gamma+\sqrt{1-\rho^2}>0$ and $|\rho|<1$. Then:
\begin{itemize}
	\item if $b(1\pm\rho)< 2$, it holds that $\tilde F(b,\rho)< 0$ and the interval $I_{\gamma,b,\rho} = \bigl]L_-(l_-(\gamma,b,\rho);\gamma,b,\rho), \newline-L_-(l_-(\gamma,b,-\rho);\gamma,b,-\rho)\bigr[$ is non-empty iff $\gamma> \tilde F(b,\rho)$;
	\item if $b(1-\rho)=2$ (or $b(1+\rho)=2$) and $\rho\neq0$, it holds that $\tilde F(b,\rho)< 0$ and the interval $I_{\gamma,b,\rho} = \bigl]-\frac{b\gamma}{2}, -L_-(l_-(\gamma,b,-\rho);\gamma,b,-\rho)\bigr[$ (resp. $I_{\gamma,b,\rho} = \bigl]L_-(l_-(\gamma,b,\rho);\gamma,b,\rho),\frac{b\gamma}{2}\bigr[$) is non-empty iff $\gamma> \tilde F(b,\rho)$;
	\item if $b=2$ and $\rho=0$, the interval $I_{\gamma,2,0}=\bigl]-\gamma,\gamma\bigr[$ is non-empty iff $\gamma>\tilde F(2,0)=0$.
\end{itemize}
In every case, the Fukasawa conditions are satisfied iff $\mu\in I_{\gamma,b,\rho}$.
\end{theorem}

	The final necessary and sufficient conditions for no Butterfly arbitrage
in the case \(\gamma+\sqrt{1-\rho^2}>0\) and \(|\rho|<1\) are summed up
in Theorem 6.2 of \cite{martini2020no}, in which \(\sigma^*\) is defined
as
\[\sigma^*(\gamma,b,\rho,\mu) := \sup_{l< l_1 \lor l>l_2}-\frac{bg_2(l)}{2G_1(l)}.\]

\begin{theorem}[Necessary and sufficient no Butterfly arbitrage conditions for SVI, $\gamma+\sqrt{1-\rho^2}>0$ and $|\rho|<1$]\label{FinalTheo}

No Butterfly arbitrage in SVI entails that $G_1$ is positive, which requires $b(1\pm\rho) \leq 2$. Under this condition:
\begin{itemize}
	\item each of the factors of the function $G_1$ is positive on $\mathbb R$  if and only if $\gamma> \tilde F(b,\rho)$ and $\mu\in I_{\gamma,b,\rho}$;
	\item for such $\mu$'s, calling $l_1< 0< l_2$ the only zeros of $g_2$, the function $G_1+\frac{1}{2\sigma}bg_2$ is positive in $]l_1,l_2[$ for every $\sigma\geq 0$ and it is non-negative on $\mathbb{R}$ if and only if $\sigma\geq\sigma^*(\gamma,b,\rho,\mu)$.
\end{itemize}

\end{theorem}

In the case of \(|\rho|= 1\) and \(\gamma\geq0\), Theorem 6.3 of
\cite{martini2020no} holds:

\begin{theorem}[Necessary and sufficient no Butterfly arbitrage conditions for SVI, $\rho=-1$]
No Butterfly arbitrage in SVI entails that $G_1$ is positive, which requires $b\leq 1$ and $\gamma \geq 0$. Under these conditions:
\begin{itemize}
	\item each of the factors of the function $G_1$ is positive on $\mathbb R$ if and only if $\mu>L_-(l_-;\gamma,b,-1)$;
	\item for such $\mu$'s, calling $l_1< 0$ the only zero of $g_2$, the function $G_1+\frac{1}{2\sigma}bg_2$ is positive on $]l_1,\infty[$ for every $\sigma\geq0$ and it is non-negative on $\mathbb{R}$ if and only if $\sigma\geq\sigma^*(\gamma,b,-1,\mu)$ where $\sigma^*(\gamma,b,-1,\mu) := \sup_{l< l_1}-\frac{bg_2(l)}{2G_1(l)}$.
\end{itemize}

\end{theorem}

From now on, we denote \(f := -\frac{bg_2}{2G_1}\) and
\(\tilde f := -\frac{G_1}{g_2(l)}\). In this way, the value of the
supremum of \(f\) is equal to \(b\) times half the reciprocal of the
infimum of \(\tilde f\), and the point at which the supremum of \(f\) is
reached is exactly the point at which the infimum of \(\tilde f\) is
reached.

	We present here some general observations that will be used in the next
sections.

In order to calculate \(\sigma^*\), one should find the supremum of the
function \(f\) over a domain constituted of two open intervals: on the
left of \(l_1\), the first zero of \(g_2\), and on the right of \(l_2\),
the second zero of \(g_2\). From a computational point of view, one
needs to perform two maximum searches and to compare the two found
values to choose the highest. Sometimes this double search is not
necessary, indeed for the sub-SVIs studied in this article, the signs of
\(\rho\) and \(\mu\) determine the interval where the global supremum of
\(f\) lies. Since \(l_1<0<l_2\), the trick will be to compare the two
quantities \(f(l)\) and \(f(-l)\), but first of all, it is necessary to
study the two intervals of interest. In particular, the following Lemma
holds:

\begin{lemma}\label{Lemmag2rho}
If $\rho\geq 0$ and $l>0$, then $g_2(l)\leq g_2(-l)$ and $l_2(\gamma,\rho)\leq-l_1(\gamma,\rho)$.

If $\rho< 0$ and $l>0$, then $g_2(l)> g_2(-l)$ and $l_1(\gamma,\rho)>-l_2(\gamma,\rho)$.
\end{lemma}

\begin{proof}

Fix $l>0$, then for $\rho\geq 0$, it holds
\begin{align*}
g_2(-l)-g_2(l) &= \frac{l\rho\Bigl((l^2+1)(1-\rho^2)+2\gamma\sqrt{l^2+1}+1\Bigr)}{(l^2+1)N(l)N(-l)}\\
&\geq \frac{l\rho\Bigl(\sqrt{(l^2+1)(1-\rho^2)}-1\Bigr)^2}{(l^2+1)N(l)N(-l)} \geq 0
\end{align*}
since $\gamma\geq-\sqrt{1-\rho^2}$. So $g_2(l)\leq g_2(-l)$ and in particular $g_2(-l_1)\leq g_2(l_1)=0$ so $l_2(\gamma,\rho)\leq-l_1(\gamma,\rho)$. These inequalities are strict for $\rho$ strictly positive. Similarly for $\rho< 0$, we find $g_2(l)> g_2(-l)$ and $l_1(\gamma,\rho)>-l_2(\gamma,\rho)$.
\end{proof}

This Lemma has a direct consequence, which is that for \(\rho\geq 0\),
the supremum of \(f\) on the left of \(l_2\) is higher than its supremum
on the left of \(-l_1\), since in the first case, \(f\) could attain it
between \(l_2\) and \(-l_1\):
\(\sup_{l>l_2} f(l) \geq \sup_{l>-l_1} f(l)\). If we can show that for
positive \(l\)s, \(f(l)\geq f(-l)\), then the last term is greater than
\(\sup_{l>-l_1} f(-l)\). Making a change of variable, this quantity is
equal to \(\sup_{l<l_1} f(l)\) so this means that \(\sigma^*\) can be
found as the supremum of \(f\) on the right of \(l_2\).

	The request \(f(l)\geq f(-l)\) is satisfied if, for example,
\(G_1(l)\leq G_1(-l)\) because in the above Lemma we showed
\(g_2(l)\leq g_2(-l)\). The difference between \(G_1(l)\) and
\(G_1(-l)\) can be written as
\((h(l)-h(-l))(h(l)+h(-l))-b^2(g(l)-g(-l))(g(l)+g(-l))\). The quantity
\(g(l)-g(-l)\) is equal to \(\frac{l}{2\sqrt{l^2+1}}\), which is
positive, while \(g(l)+g(-l)\) is \(\frac{\rho}{2}\), which again is
non-negative for \(\rho\geq0\). Showing that the product with the \(h\)
functions is non-positive would then be enough to reach the desired
inequality.

\begin{lemma}\label{Lemmafrho0}

Let $\rho=0$. If $\mu\geq 0$, then $\sigma^*(\gamma,b,0,\mu) = \sup_{l> l_2}f(l)$ while if $\mu< 0$, then $\sigma^*(\gamma,b,0,\mu) = \sup_{l< l_1}f(l)$.

\end{lemma}

\begin{proof}

In the case $\rho=0$, it holds $l_2=-l_1$ since $g_2$ is symmetric. It can be shown that in general
\begin{align*}
h(l)-h(-l) &= -\frac{l}{\sqrt{l^2+1}}\frac{\rho\bigl(\gamma\sqrt{l^2+1}+1\bigr)+\mu\bigl(\gamma+(1-\rho^2)\sqrt{l^2+1}\bigr)}{N(l)N(-l)},\\
h(l)+h(-l) &= \frac{\bigl(\gamma\sqrt{l^2+1}+1\bigr)\bigl(2\gamma+2\sqrt{l^2+1}-\rho\mu\bigr) + l^2\bigl(\gamma+(1-\rho^2)\sqrt{l^2+1}\bigr)}{\sqrt{l^2+1}N(l)N(-l)}.
\end{align*}

In particular for $\rho=0$ and positive $l$, the sign of the former quantity is the sign of $-\mu$ while the numerator of the latter quantity is $(\gamma+\sqrt{l^2+1})(2\gamma\sqrt{l^2+1}+l^2+2)$, which is always positive for $\gamma>-1$. Then $G_1(l)-G_1(-l)=(h(l)-h(-l))(h(l)+h(-l))$ follows the sign of $-\mu$ and so does $f(-l)-f(l)$. If $\mu$ is non-negative, $f(l)\geq f(-l)$ so its supremum must be searched on the right of $l_2$. If $\mu$ is negative, the opposite inequality holds for $f$ and its supremum lies on the left of $l_1$.

\end{proof}

These Lemmas will be useful for the sub-SVIs studied in this article,
however it must be noticed that the inequality \(f(l)\geq f(-l)\) does
not hold in general.

\subsection{Smile inversion}\label{smile-inversion}

	Tehranchi proved in \cite{tehranchi2020black} that a curve of Call
prices (with unit underlyer) parametrized by the strike \(\kappa\) is
free of Butterfly arbitrage if and only if it is convex and satisfies
\(1 \geq C(\kappa)\geq (1-\kappa)^+\) for every \(\kappa\geq0\).
Moreover, \(C\) has these properties if and only if
\(C^*(\kappa):=1-\kappa+\kappa C\bigl(\frac{1}{\kappa}\bigr)\) has them.

	Note that if at some point $\kappa$ it holds $C(\kappa)=1$, then since $C$ should be convex and decreasing, then $C$ is the constant function $1$. Because this case is uninteresting, we can assume that \(1>C(\kappa)\) for every \(\kappa\). Then
there is a unique total variance function \(w(k)\) such that
\(C_{BS}(\kappa,\sqrt{w(k)})=C(\kappa)\) where \(C_{BS}\) is the
(normalized) Black-Scholes Call price function, and a unique \(w^*(k)\)
such that \(C_{BS}(\kappa,\sqrt{w^*(k)})=C^*(\kappa)\). By the Put-Call
parity for the Black-Scholes model it holds that
\[P_{BS}(\kappa,\sqrt{w^*(k)})= C_{BS}(\kappa,\sqrt{w^*(k)})+\kappa-1 = \kappa C\Bigl(\frac{1}{\kappa}\Bigr).\]

	Now the LHS is equal to
\(\kappa N\Bigl(\frac{k}{\sqrt{w^*(k)}} + \frac{\sqrt{w^*(k)}}2\Bigr) - N\Bigl(\frac{k}{\sqrt{w^*(k)}} - \frac{\sqrt{w^*(k)}}2\Bigr)\)
and the RHS is equal to
\(\kappa \Bigl( N\Bigl(\frac{k}{\sqrt{w(-k)}} + \frac{\sqrt{w(-k)}}2\Bigr) - \frac{1}{\kappa}N\Bigl(\frac{k}{\sqrt{w(-k)}} - \frac{\sqrt{w(-k)}}2\Bigr) \Bigr)\).
By the monotonicity of the function \(u \to C_{BS}(\kappa, u)\) it
follows then that \(w^*(k)=w(-k)\). Eventually, we can reword the
symmetry of Tehranchi involution \(C \to C^*\) at the implied volatility
level: a smile \(w\) is free of Butterfly arbitrage if and only if the
inverse smile \(k \to w(-k)\) is.

	Let us now relate the implied volatility of \(C\) with that of \(C^*\).

	In the case of SVI, this reduces to have
\begin{equation}\label{eqSVIstar}
a+b\Bigl(\rho(k-m) + \sqrt{(k-m)^2+\sigma^2}\Bigr) = a^*+b^*\Bigl(\rho^*(-k-m^*) + \sqrt{(-k-m^*)^2+\sigma^{*2}}\Bigr)
\end{equation} for all \(k\). We prove that this happens iff \(a^*=a\),
\(b^*=b\), \(\rho^*=-\rho\), \(m^*=-m\) and \(\sigma^*=\sigma\). Of
course, \(a^*\), \(b^*\), \(\rho^*\), \(m^*\) and \(\sigma^*\) do not
depend on \(k^*\) so neither on \(k\). Making the derivative with
respect to \(k\) at the latter equality gives
\begin{equation}\label{eqSVIstarder}
b\biggl(\rho + \frac{k-m}{\sqrt{(k-m)^2+\sigma^2}}\biggr) = b^*\biggl(-\rho^* + \frac{k+m^*}{\sqrt{(k+m^*)^2+\sigma^{*2}}}\biggr).
\end{equation} Evaluating this for \(k\) going to \(\pm\infty\), it must
hold \(b(\rho\pm 1)=b^*(-\rho^*\pm 1)\) or equivalently \(b^*=b\) and
\(\rho^*=-\rho\). In the same equation, consider \(k=m\), then it must
hold
\[b\rho = b\biggl(\rho + \frac{m+m^*}{\sqrt{(m+m^*)^2+\sigma^{*2}}}\biggr)\]
which reduces to \(m^*=-m\). At this point, (\ref{eqSVIstarder}) becomes
\(\frac{k-m}{\sqrt{(k-m)^2+\sigma^2}} = \frac{k+m^*}{\sqrt{(k+m^*)^2+\sigma^{*2}}}\)
so \(\sigma^*=\sigma\). Finally, from (\ref{eqSVIstar}), it follows
\(a^*=a\).

	\begin{proposition}[Absence of Butterfly arbitrage for the inverse SVI]\label{PrepInverseSVI}

The $SVI(a,b,\rho,m,\sigma)$ is Butterfly arbitrage-free iff the $SVI(a,b,-\rho,-m,\sigma)$ is Butterfly arbitrage-free.

\end{proposition}

\subsection{The $b^*$ approach: reparametrizing from
the critical point
equation}\label{the-b-approach-reparametrizing-from-the-critical-point-equation}

	We now explain how to properly perform the trick described in the
Introduction \ref{introduction} which allows to switch the roles of the
parameter \(b\) and of the critical point \(\bar{l}\) depending on \(b\)
of \(\tilde f\). The key passage will be the exploitation of the
characteristic equation of these critical points:
\(h(\bar l)p(\bar l)-b^2 g(\bar l)q(\bar l)=0\). Three hypothesis must
be verified:

\begin{enumerate}
\def\labelenumi{\arabic{enumi}.}
\item
  the Fukasawa bound for \(\gamma\) must be a monotone function in
  \(b\);
\item
  there is uniqueness for the critical points of \(\tilde f\);
\item
  the functions \(p\) and \(q\) do not vanish at the same point.
\end{enumerate}

	In \cite{martini2020no}, the Fukasawa conditions read
\(\gamma>\tilde F(b,\rho)\). Firstly, suppose that for fixed \(\rho\),
\textbf{the function \(b\to \tilde F(b,\rho)\) is monotone and
surjective from \(\bigl[0,\frac{2}{1+|\rho|}\bigr]\) to \([-1,0]\)}.
Then, its inverse \(\tilde F^{-1}(\cdot,\rho)\) is well defined on
\([-1,0]\) and we shall extend it to \begin{equation*}
\tilde G(\gamma,\rho) :=
\begin{cases}
\tilde F^{-1}(\gamma,\rho) & \text{if}\ \gamma\in]-1,0],\\
\frac{2}{1+|\rho|} & \text{if}\ \gamma>0.
\end{cases}
\end{equation*}

	Secondly, suppose we have proven that \textbf{for \(l>l_2\), the
function \(\tilde f\) has exactly one critical point
\(\bar l(\gamma,b,\rho,\mu)\) for
\(\rho\in[-1,1], \gamma>-1, b\in[0,\tilde G(\gamma,\rho)], \mu\in]L_-,L_+[\)}.
This critical point must then be a local and global point of minimum.
Note that we require the uniqueness also for \(b=0\) and
\(b=\tilde G(\gamma,\rho)\). This point satisfies
\(\tilde f'(\bar l)=0\) or equivalently
\[h(\bar l)p(\bar l)-b^2 g(\bar l)q(\bar l)=0\] with
\(h(\bar l)= G_{1+}(l)+G_{1-}(l) > 0\), \(g(l)>0\) and
\(p(l):=h(l)g_2'(l) - 2h'(l)g_2(l)\) and
\(q(l):=g(l)g_2'(l)-2g'(l)g_2(l)\).

In particular, either \(q(\bar l)=p(\bar l)=0\) or
\[b^2=\frac{h(\bar l)\bigl(h(\bar l)g_2'(\bar l) - 2h'(\bar l)g_2(\bar l)\bigr)}{g(\bar l)\bigl(g(\bar l)g_2'(\bar l)-2g'(\bar l)g_2(\bar l)\bigr)}\]
when \(q(\bar l) \neq 0\).

It is natural to define on the set
\(B(\gamma,\rho,\mu) = \{l\in]l_2(\gamma,\rho),\infty[: p(l)q(l)>0\}\) a
positive function \(b^*\) by the formula \begin{equation}\label{eqbstar}
b^{*2}(l):=\frac{h(l)p(l)}{g(l)q(l)}.
\end{equation}

Note that the function \(b^*\) might go on \(B(\gamma,\rho,\mu)\) to
levels not allowed for \(b\); yet this function is continuous on its
domain of definition and when \(q(\bar l) \neq 0\) it holds that
\(b^2 = b^{*2}(\bar l)\). As an immediate consequence, note that if
\(b_1 \neq b_2\) are such that
\(\bar l(\gamma,b_1,\rho,\mu), \bar l(\gamma,b_2,\rho,\mu) \in B(\gamma,\rho,\mu)\),
then \(\bar l(\gamma,b_1,\rho,\mu) \neq \bar l(\gamma,b_2,\rho,\mu)\).
Note also that we don't know if \(B\) is one-piece, i.e.~connected.

	The question of interest is now the location of the set
\(Z(\gamma,\rho,\mu):=\{\bar l(\gamma,b,\rho,\mu), b\in]0,\tilde G(\gamma,\rho)[\}\)
(note that we exclude \(0\), this will turn to be more convenient below)
with respect to the sets \(B(\gamma,\rho,\mu)\) and \(\{p=q=0\}\). The
third and last hypothesis is that, \textbf{for the fixed parameters
\((\gamma,\rho,\mu)\), the set \(\{p=q=0\}\) is empty}. Then
\(Z(\gamma,\rho,\mu)\) is contained in \(B(\gamma,\rho,\mu)\). From the
above remark, the function \(\bar l(\gamma,\cdot,\rho,\mu)\) is
injective; however at this stage we don't know whether it is continuous
(which would imply it is either continuous increasing or continuous
decreasing).

	Remember that the proof of the uniqueness of the critical point of
\(\tilde f\) still holds for \(b= \tilde G(\gamma,\rho)\). Then for
\((l,b) \in ]l_{2},\infty] \times [0, \tilde G(\gamma,\rho)]\), the
equation \[h(l)p(l)-b^2 g(l)q(l)=0\] characterizes the points
\(\bar l(\gamma,b,\rho,\mu)\), since \(\tilde f\) has a single local and
global minimum. As discussed above, in the open set
\(B(\gamma,\rho,\mu)\), this equation defines a continuous function
\(b^*\); from the characterizing property we get that
\(Z(\gamma,\rho,\mu)\) eventually coincides with
\(b^{*-1}(]0, \tilde G(\gamma,\rho)[)\subseteq B(\gamma,\rho,\mu)\), so
that in particular \(Z(\gamma,\rho,\mu)\) is an open set. Furthermore,
the function \(b \to \bar l(\gamma,b,\rho,\mu)\) is the inverse of
\(b^*:Z(\gamma,\rho,\mu)\to]0,\tilde G(\gamma,\rho)[\).

	It remains to prove that \(Z(\gamma,\rho,\mu)\) is \emph{one-piece}.
Indeed, take a sequence \((b_n)_n\) such that
\(\bar l_n=\bar l(\gamma,b_n,\rho,\mu)\) goes to
\(\bar l\in\bar Z(\gamma,\rho,\mu)\). Since \((b_n)_n\) is a bounded
sequence, it has a subsequence converging to a certain
\(b\in[0,\tilde G(\gamma,\rho)]\); call such subsequence as the original
one, then the correlated subsequence of \(\bar l_n\) still converges to
\(\bar l\). Since for every \(n\) we have
\(h(\bar l_n)p(\bar l_n)-b_n^2g(\bar l_n)q(\bar l_n)=0\) where all the
functions are continuous, then taking the limit, it holds
\(h(\bar l)p(\bar l)-b^2g(\bar l)q(\bar l)=0\). From the characteristic
equation above, this yields in turn that
\(\bar l=\bar l(\gamma,b,\rho,\mu)\). Also, \(b\) is unique because
either \(b=b^*(\bar l)\) or \(p(\bar l)=q(\bar l)=0\), but the set
\(\{p=q=0\}\) is empty. Then it must hold either that
\(\bar l\in Z(\gamma,\rho,\mu)\) or that \(\bar l\) is the unique
critical point \(\bar l(\gamma,b,\rho,\mu)\) of \(\tilde f\) when
\(b=0\) or \(b=\tilde G(\gamma,\rho)\). The boundary of
\(Z(\gamma,\rho,\mu)\) is the set
\(\{\bar l(\gamma,0,\rho,\mu),\bar l(\gamma,\tilde G(\gamma,\rho),\rho,\mu)\}\)
and because \(Z(\gamma,\rho,\mu)\) is an open set, we get that
\(Z(\gamma,\rho,\mu)=]\bar l(\gamma,0,\rho,\mu),\bar l(\gamma,\tilde G(\gamma,\rho),\rho,\mu)[\)
or
\(Z(\gamma,\rho,\mu)=]\bar l(\gamma,\tilde G(\gamma,\rho),\rho,\mu),\bar l(\gamma,0,\rho,\mu)[\).

	As a consequence we also get that \(\bar l(\gamma,\cdot,\rho,\mu)\) is
either strictly increasing or strictly decreasing, and that
\(Z(\gamma,\rho,\mu)\) is a single connected component of
\(B(\gamma,\rho,\mu)\), with a boundary point
\(\bar l(\gamma,0,\rho,\mu)\) which is the single zero of \(p\), and the
other boundary point lying within \(B(\gamma,\rho,\mu)\).

\section{Vanishing SVI}\label{vanishing-svi}

\subsection{Vanishing (Upward) SVI}\label{vanishing-upward-svi}
	In this section, we work with the Vanishing Upward SVI and immediately recover the final results on the Vanishing Downward SVI in \Cref{vanishing-downward-svi}.
	
	The Vanishing Upward SVI is the sub-SVI obtained by setting \(\rho=1\)
and \(a=0\). The corresponding SVI formula becomes \begin{equation*}
SVI(k;0,b,1,m,\sigma) = b(k-m + \sqrt{(k-m)^2+\sigma^2}).
\end{equation*} With our notations, \(N(l) = l+\sqrt{l^2+1}\). Note that
the Roger Lee conditions require \(0<b \leq 1\).

	We plot in \Cref{FigureVanishingUpward} a Vanishing Upward SVI with \(b=\frac{1}{2}\), \(m=-1\) and
\(\sigma=1\).

\begin{figure}
	\centering
	\includegraphics[width=.7\textwidth]{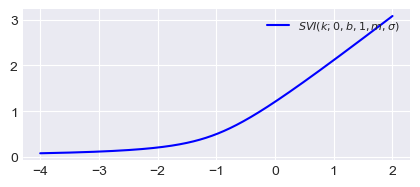}
	\caption{Vanishing Upward SVI with $b=\frac{1}{2}$, $m=-1$ and $\sigma=1$.}
	\label{FigureVanishingUpward}
\end{figure}
    
	The wording \emph{Vanishing} refers to the fact that the smile goes to 0
on the left, \emph{upward} meaning it is increasing. The symmetric
smile with \(\rho=-1\) will be a \emph{Vanishing Downward} one.

\subsection{The Fukasawa conditions}\label{the-fukasawa-conditions}

	Since \(\gamma=0\), the Fukasawa condition on \(\gamma\) is
automatically satisfied. We cite here a result obtained in paragraph
5.3.1 of \cite{martini2020no}.

\begin{lemma}[Fukasawa conditions for the Vanishing Upward SVI]

A Vanishing Upward SVI with $0< b\leq 1$ satisfies the Fukasawa conditions iff $\mu<\sqrt{3(1-b)}$.

\end{lemma}

\subsection{The condition on $\sigma$}\label{the-condition-on-sigma}

	Let us define the two following auxiliary functions: \begin{align}
\mu^*(x) :=& \Bigl[2(1-x)(2x^{2}-8x-1) +\nonumber\\
& +\sqrt{4b^2x^{6} + 8b^2x^{5} + 8x^{4}(8-b^2) - 4x^{3}(5b^2+32) + x^{2}(96 - b^2) + 2x(5b^2-16) + 4+3b^2}\Bigr]\nonumber\\
&/\Bigl[2\sqrt{1-x^{2}}(2x^{2}-2x-1)\Bigr], \label{eqVanishingmustar}\\
\sigma^*(x) :=& -\frac{4b\sqrt{1-x^{2}}(1-x -2x^{2})}{4\bigl(2-x-\mu^*(x) \sqrt{1-x^{2}}\bigr)^2-b^2(1+x)^2} \label{eqVanishingsigmastar}
\end{align}

We show then the following:

\begin{proposition}[Fully explicit no arbitrage domain for the Vanishing Upward SVI]\label{PropVanishingUp}
A Vanishing Upward SVI with $b=1$ is arbitrage-free iff $\mu< 0$ and $\sigma\geq-\frac{\mu}{2}$.

A Vanishing Upward SVI with $0< b< 1$ is arbitrage-free iff it can be parametrized as
\begin{equation}\label{eqSVINewParamVanishingUp}
SVI(k) = b\sigma\Biggl(\frac{k}\sigma-\mu^*(x) + \sqrt{\biggl(\frac{k}\sigma-\mu^*(x)\biggr)^2+1}\Biggr)
\end{equation}
where $\frac{2+b}{4-b}< x< 1$ and $\sigma\geq \sigma^*(x)$.
\end{proposition}

\subsubsection{Proof of \Cref{PropVanishingUp}}\label{proof-of}

	Note that \begin{align*}
G_1(l) &= \biggl(1-\frac{(l+\mu)}{2\sqrt{l^2+1}}\biggr)^2-\frac{b^2}{16}\biggl(1+\frac{l}{\sqrt{l^2+1}}\biggr)^2,\\
g_2(l) &= \frac{1}{(l^2+1)^\frac{3}{2}} - \frac{\sqrt{l^2+1} + l}{2(l^2+1)}.
\end{align*} The \(g_2\) function has only one positive zero. Indeed,
its zeros solve \(2=l^2+1+l \sqrt{l^2+1}\), or yet
\(1-l^2=l\sqrt{l^2+1}\). The only possible solution satisfies
\(l^2 \leq 1\), so \(l_2 = \frac{1}{\sqrt{3}}\). In order to have no
arbitrage, we then need
\(\sigma\geq\sigma^*=\sup_{l\in]l_2,+\infty[}-\frac{bg_2(l)}{2G_1(l)}\).

	From now on, we operate a change of variable setting
\(x=\frac{l}{\sqrt{l^2+1}}\) so that the points \(l=l_2\) and
\(l=\infty\) correspond to \(x=\frac{1}2\) and \(x=1\) respectively.
Also, \(\frac{1}{\sqrt{l^2+1}}=\sqrt{1-x^2}\).

We call \(J_1\) and \(j_2\) the functions \(G_1\) and \(g_2\) evaluated
at \(x\), so \begin{align*}
J_1(x) &= \biggl(1-\frac{x}2-\frac{\mu}2 \sqrt{1-x^2}\biggl)^2-\frac{b^2}{16}(1+x)^2,\\
j_2(x) &= \frac{\sqrt{1-x^2}}{2} (1-x -2x^2).
\end{align*}

	The first derivative of \(j_2\) is
\(j_2'(x) = \frac{1}{2\sqrt{1-x^2}}(6x^3+2x^2-5x-1)\) which is positive
iff \(x>\frac{2+\sqrt{10}}{6}:=x_{m_2}\). Also,
\(j_2''(x) = \frac{-12x^4-2x^3+18x^2+3x-5}{2(1-x^2)^\frac{3}2}\) and the
only inflection point of \(j_2\) in \(\bigl[\frac{1}2,1\bigr]\) is
\(\frac{1}2\). Then, the function \(j_2\) is null at \(\frac{1}2\),
decreases to its minimum
\(j_2(x_{m_2})=-\frac{34\sqrt{2}-5\sqrt{5}}{108}\) and then increases to
\(0\) for \(x=1\), furthermore it is always convex.

	Let us turn to the study of \(J_1\). We have
\(J_1(1) = \frac{1-b^2}{4}\) and we know that for the Fukasawa
conditions, both \(J_{1+}\) and \(J_{1-}\) are positive, where
\(J_1=J_{1+}J_{1-}\) and
\(J_{1\pm} = 1-\frac{x}2 -\frac{\mu}2 \sqrt{1-x^2} \mp \frac{b}4 (1+x)\).
Considering the first derivatives, we know
\(J_1'= J_{1+}'J_{1-} + J_{1+}J_{1-}'\) where
\begin{equation}\label{J1Diffeq}
J_{1\pm}'(x) = -\frac{1}2 + \frac{\mu}{2}\frac{x}{\sqrt{1-x^2}}\mp\frac{b}4.
\end{equation}

If \(\mu\leq 0\), then \(J_{1\pm}\) are decreasing and so is \(J_1\). In
particular, it attains its minimum at \(x=1\) and we get that
\(J = J_1 + \frac{1}{2\sigma}bj_2\) is always positive for
\[\sigma \geq -\frac{bj_2(x_{m_2})}{2J_1(1)}=\frac{(34\sqrt{2}-5\sqrt{5})b}{54(1-b^2)}\]
when \(b<1\). We have therefore the explicit no arbitrage sub-domain:

	\begin{lemma}[No arbitrage sub-domain for the Vanishing Upward SVI]\label{LemmaVanishingSubDomain}

A Vanishing Upward SVI with $0< b < 1$, $\mu\leq 0$ and $\sigma\geq\frac{(34\sqrt{2}-5\sqrt{5})b}{54(1-b^2)}$ is arbitrage-free.

\end{lemma}

	This explicit sub-domain was obtained with not too much effort. Let us
turn now to the more difficult task to obtain an explicit
parametrization for the whole domain.

\paragraph{Uniqueness of the critical point of $\tilde f$ for the Vanishing
SVI}\label{uniqueness-of-the-critical-point-of-tilde-f-for-the-vanishing-svi}

	Let us study \(\tilde \phi(x;b,\mu) = -\frac{J_1(x;b,\mu)}{j_2(x;b)}\).
We want to prove that for each \(b\in[0,1]\) and
\(\mu\in]-\infty,\sqrt{3(1-b)}[\), there exists a unique value
\(x^*:=x^*(b,\mu)\in\bigl[\frac{1}2,1\bigr]\) such that
\(\inf_{x\in]\frac{1}2,1[}\tilde \phi(x) = \tilde \phi(x^*)\).

	The existence is obvious. The derivative of \(\tilde \phi\) is
\begin{align}\label{Derfeq}
\tilde{\phi}'(x;b,\mu) =& \Bigl[4\mu^{2}(1-x^2)(2x^2-2x-1) -8\mu  (1-x)\sqrt{1 - x^{2}}(2x^2-8x-1) +\nonumber\\
&+2x^{4} (4-b^2) - 6x^{3}(b^{2}+12) + 3x^{2}(52 - b^{2}) + 4x(b^{2} - 22) + 3 b^{2}\Bigr]\nonumber\\
&/\Bigl[8(1-x^2)^\frac{3}{2}(1+x)(2x-1)^2\Bigr].
\end{align} The denominator of the above formula is always positive.
Looking at the numerator, for \(b=1\) every coefficient of \(\mu\) is
negative when \(x\in\bigl]\frac{1}2,1\bigr[\) and \(\mu<0\), required
from the Fukasawa condition. Then \(\tilde{\phi}'(x;1,\mu)\) is always
negative and the inferior point of \(\tilde \phi\) is reached at \(1\)
for every \(\mu< 0\), so that \(x^*(1,\mu)=1\). From now on we consider
\(b< 1\). Since for finite values of \(\mu\),
\(\tilde \phi(\frac{1}2;b,\mu) = \tilde \phi(1;b,\mu) = \infty\), the
points which attain the infimum of \(\tilde \phi\) are points of minimum
belonging to \(\bigl]\frac{1}2,1\bigr[\) and such that
\(\tilde \phi'(x^*;b,\mu)=0\).

	What happens when \(\mu=-\infty\)? Dividing by \(\mu^2\) we still have
\(\frac{\tilde \phi'(x^*;b,\mu)}{\mu^2}=0\), so making \(\mu\) going to
\(-\infty\) and relying on the linearity of the limits, we obtain the
equation \(4(1-x^{*2})(2x^{*2}-2x^*-1)=0\), whose only solution in
\(\bigl[\frac{1}{2},1\bigr]\) is \(1\), so \(x^*(b,-\infty)=1\).

Take now \(\mu\leq0\) and \(|\nu|<-\mu\) (note in particular that
\(\mu < \nu\)) and consider the quantity
\(\tilde \phi(x;b,\mu)-\tilde \phi(x;b,\nu)\). Unless a factor
\(b^{-1}(1+x)^{-1}(2x-1)^{-1}\), it is equal to
\((\nu-\mu) \bigl(2-x-\frac{(\mu+\nu)}2 \sqrt{1-x^2}\bigr)\) and, as a
consequence, the quantity
\(\tilde \phi'(x;b,\mu)- \tilde \phi'(x;b,\nu)\) times \((1+x)(2x-1)\)
equals
\[(\nu-\mu) \biggl(-1+\frac{(\mu+\nu)}2 \frac{x}{\sqrt{1-x^2}}\biggr) + (4x+1) \bigl(\tilde \phi(x;b,\nu)-\tilde \phi(x;b,\mu) \bigr).\]
The function \(\tilde \phi\) is decreasing in \(\mu\), indeed
\(\partial_\mu\tilde \phi(x) = -\frac{2-x-\mu\sqrt{1-x^2}}{b(1+x)(2x-1)}\)
and \(2-x-\mu\sqrt{1-x^2}>2-x-\sqrt{3}\sqrt{1-x^2}\geq0\). So given the
fact that \(\tilde \phi\) is decreasing in \(\mu\), we get in turn
\(\tilde \phi'(x;b,\mu)- \tilde \phi'(x;b,\nu)< 0\). This entails that
if there is uniqueness, \(x^*(b,\nu)< x^*(b,\mu)\) under those
conditions.

	We prove that indeed the uniqueness holds. We prove that if
\(\tilde \phi'\geq 0\), then \(\tilde \phi''>0\), which means that once
\(\tilde \phi'\) becomes zero, then it will necessarily increase and it
can never become \(0\) again (note that this property is slightly weaker
than the convexity of \(\tilde\phi\)). We have
\(\tilde \phi' = \frac{J_1j_2'-J_1'j_2}{j_2^2}\geq 0\) but the term
\(-J_1'j_2\) is negative, so \(j_2'\) must be positive. It holds
\(\tilde \phi'' = \frac{J_1j_2''-J_1''j_2}{j_2^2}-2\tilde \phi'\frac{j_2'}{j_2}\),
where we have proven \(-2\tilde \phi'\frac{j_2'}{j_2}\geq 0\). We now
look at the quantity \(n = J_1j_2''-J_1''j_2\) and show that it is
strictly positive to obtain the conclusion. Indeed, the denominator of
\(n\) is \(32\sqrt{1 - x^{2}}(1-x)\) which is positive while its
numerator is a quadratic function of \(\mu\) with quadratic coefficient
\(-4(1-x^2)(2x-1)(4x^2-2x-3)\), linear coefficient
\(16\sqrt{1 - x^2}(1-x)(2x-1)(2x^2-5x-4)\) and free term depending on
\(b\). The quadratic coefficient is positive, so \(n\) is a convex
parabola as a function of \(\mu\). Furthermore, the linear coefficient,
corresponding to \(\partial_\mu n(x;\mu)|_{\mu=0}\) unless a positive
factor, is negative, so \(n(x;0) < n(x;\mu)\) for every \(\mu< 0\). Then
for negative \(\mu\)'s it is enough to prove that \(n(x;0)\) is positive
for every \(x\). We already know that \(j_2''\) is positive. In
addition, \(J_1'' = J_{1+}''J_{1-}+2J_{1+}'J_{1-}'+J_{1+}J_{1-}''\)
where \(J_{1\pm}''(x;0)=\frac{\mu}{2(1-x^2)^\frac{3}2}\Bigr|_{\mu=0}=0\)
and \(J_{1\pm}'(x;0)< 0\). Then \(J_1''(x;0)>0\) and
\(n(x;0)=J_1(x;0)j_2''(x;0)-J_1''(x;0)j_2(x;0)>0\).

	For the final case \(\mu>0\), we need to introduce the
\(\mu^*\)-function.

\subparagraph{The \(\mu^*\) approach for the Vanishing SVI}

	Assume \(b< 1\) so that \(J_1(1)>0\) and let us make \(\mu\) vary
towards its upper bound defined by the Fukasawa condition. Then the
function \(J_1\) will eventually reach the x-axis level at a point
\(x^*_+(b)\). This point is necessarily a critical point of \(J_1\) with
\(\mu\) set at the value of the upper bound. So
\(J_1(x^*_+(b))=J_1'(x^*_+(b))=0\). Now observe that
\(\tilde \phi'=\frac{J_1j_2'-J_1'j_2}{j_2^2}\) so that
\(\tilde \phi'(x^*_+(b);b,\sqrt{3(1-b)})=0\) also. In particular, if we
know that \(\tilde \phi\) (for this critical value of \(\mu\)) has a
single critical point, this must be \(x^*_+(b)\).

	We look for the solutions to \(J_1(x)=0\) when \(\mu = \sqrt{3 (1-b)}\).
\(J_{1-}\) is always positive on \(\bigl]\frac{1}2, 1\bigr[\) while
\(J_{1+}(x)=1-\frac{x}2 -\frac{\mu}2 \sqrt{1-x^2} - \frac{b}4 (1+x)\)
satisfies
\(\bigl(\bigl(1-\frac{b}4\bigr)x-\bigl(\frac{1}2+\frac{b}4\bigr) \bigr)^2=0\).
This yields
\[x_+^*(b)=\frac{\frac{1}2+\frac{b}4}{1-\frac{b}4}=\frac{2+b}{4-b}\]
which is clearly larger than \(\frac{1}2\) for positive \(b\) and
smaller than \(1\) when \(b< 1\).

	As a sanity check we should verify that also \(J_1'(x)=0\) at this
point. Indeed,
\(J_1'(x_+^*(b)) = J_{1+}'(x_+^*(b))J_{1-}(x_+^*(b)) + J_{1+}(x_+^*(b))J_{1-}'(x_+^*(b)) = J_{1+}'(x_+^*(b))J_{1-}(x_+^*(b))\).
From \cref{J1Diffeq} we have that \(J_{1+}'\) is null for
\(\mu=\sqrt{3(1-b)}\) and \(x=x_+^*(b)\).

	In the limiting case \(b=1\) and \(\mu=0\), we have already seen that
\(J_1\) is decreasing and attains its minimum at \(1\). Furthermore,
\(\tilde \phi(1;1,0)=0\), so also \(\tilde \phi\) reaches its minimum at
\(1\). We can then set \(x^*(1)=1\). The uniqueness of \(x^*\) follows
from the fact that \(J_1\) and so \(\tilde\phi\) are positive in
\(\bigl[\frac{1}2,1\bigr[\).

	Take \(x\in]x^*_+(b),1[\). From \cref{Derfeq} we can see that
\(\tilde \phi'\) is a concave parabola as a function of \(\mu\). So in
order to have a unique \(\mu = \mu^*(x)\)-solution to
\(\tilde \phi'(x;b,\mu)=0\), it is enough to prove
\(\tilde \phi'(x;b,\sqrt{3(1-b)})>0\). Note that \(x^*_+(b)\) is a zero
and a point of minimum for \(\tilde \phi(\cdot;b,\sqrt{3(1-b)})\), so in
a right neighborhood of \(x^*_+(b)\) we have
\(\tilde \phi'(x;b,\sqrt{3(1-b)})>0\). If for a certain \(x\) we rather
have \(\tilde \phi'(x;b,\sqrt{3(1-b)})< 0\), then there exists a
\(x^*_-(b)>x^*_+(b)\) point of maximum such that
\(\tilde \phi'(x^*_-(b);b,\sqrt{3(1-b)})=0\) and
\(\tilde \phi''(x^*_-(b);b,\sqrt{3(1-b)})< 0\). We show that this is not
possible.

	Let us look at \(J_1'\). From \cref{J1Diffeq}, we see that
\(J_{1+}'(x)\) is positive iff
\(\frac{x}{\sqrt{1-x^2}}>\frac{b+2}{2\mu}\) or
\(x>\frac{2+b}{\sqrt{4\mu^2+(2+b)^2}}\). For \(\mu=\sqrt{3(1-b)}\) we
obtain exactly that \(x\) has to be greater than \(x^*_+(b)\). On the
other side, \(J_{1-}'(x)\) is positive iff
\(x>\frac{2-b}{\sqrt{4\mu^2+(2-b)^2}}\) which corresponds to
\(x>\frac{2-b}{\sqrt{b^2-16b+16}}\) for \(\mu=\sqrt{3(1-b)}\). This
quantity is inferior to \(x^*_+(b)\) so also \(J_{1-}'(x)\) is positive.
Then \(J_1'(x^*_-(b))>0\). Suppose \(x^*_-(b)\geq x_{m_2}\), then
\(j_2'(x^*_-(b))\geq0\) and consequently \(\tilde \phi'(x^*_-(b))>0\),
which is not possible. So if it exists, \(x^*_-(b)\) must be smaller
than \(x_{m_2}\).

	We have
\(\tilde \phi'' = \frac{J_1j_2''-J_1''j_2}{j_2^2}-2\tilde \phi'\frac{j_2'}{j_2}\)
so when evaluating in \(x^*_-(b)\), the second term is null. The
quantities \(J_1(x^*_-(b))\), \(-j_2(x^*_-(b))\) and \(j_2''(x^*_-(b))\)
are strictly positive. Also,
\(J_1''=J_{1+}''J_{1-}+2J_{1+}'J_{1-}'+J_{1+}J_{1-}''\) and
\(J_{1\pm}''(x)=\frac{\mu}{2(1-x^2)^\frac{3}2}\) which are positive.
Since also \(J_{1\pm}\) and \(J_{1\pm}'\) are positive in \(x^*_-(b)\),
then \(J_1''(x^*_-(b))>0\). So
\(\tilde\phi''(x^*_-(b);b,\sqrt{3(1-b)})\) cannot be negative and this
leads to a contradiction.

	Consequently, for fixed \(b\in]0,1[\), there is a unique function
\(x\to \mu^*(x)\) such that \(\tilde{\phi}'(x;b,\mu^*(x))=0\) is well
defined for \(x\in]x^*_+(b),1[\) and takes values in
\(]-\infty,\sqrt{3(1-b)}[\). It is also continuous since it is defined
as the first root of a second degree polynomial with continuously
changing parameters. In particular, \(\mu^*\) is defined as in
\cref{eqVanishingmustar}.

Note that the function \(\mu^*\) cannot be extended to \(b=1\) since the
inferior point of \(\tilde \phi\) is reached at \(1\) for every
\(\mu< 0\). On the other hand, it can be extended to \(b=0\) since all
the previous statements still hold.

	We have already seen that for \(x\) in the preimage of \(]-\infty,0]\),
the function \(\mu^*(x)\) is injective. We show that this still holds in
the preimage of \(]0,\sqrt{3(1-b)}[\). If this is not the case, there
exists a critical point \(\hat x\) such that \({\mu^*}'(\hat x) = 0\).
Taking the derivative with respect to \(x\) at the members of
\(\tilde\phi'(x;\mu^*(x))=0\), we obtain
\(\tilde\phi''(x;\mu^*(x)) + \partial_\mu\tilde\phi'(x;\mu^*(x)){\mu^*}'(x)=0\),
so at \(\hat x\) it holds simultaneously
\(\tilde\phi'(\hat x;\mu^*(\hat x))=0\) and
\(\tilde\phi''(\hat x;\mu^*(\hat x))=0\). We want to prove that this is
not possible.

At \(\hat x\), we have
\(\tilde\phi''\mid_{\mu=\mu^*(\hat x)} = \frac{J_1j_2''-J_1''j_2}{j_2^2}\mid_{\mu=\mu^*(\hat x)}\).
Since we have already proven that \(j_2''>0\) on
\(\bigl]\frac{1}{2},1\bigr[\), it is enough to prove that \(J_1''>0\).
We consider \(\mu\geq0\). It can be shown that
\(J_1''(x) = \frac{(4- b^{2}-4\mu^2)(1-x^2)^2 + 4\mu\sqrt{1 - x^{2}}(2x^{3} - 3x + 2)}{8(1-x^{2})^2}\)
and \(J_1'''(x) = \frac{3 \mu (2 x - 1)}{2(1 - x^{2})^\frac{5}2}\),
which is positive iff \(x>\frac{1}2\). So \(J_1''\) is increasing. We
have
\(J_1''\bigl(\frac{1}{2}\bigr)=-\frac{b^2}8-\frac{\mu^2}2 + \frac{\mu}{\sqrt{3}} + \frac{1}2\)
which is positive iff \(-4\mu^2+\frac{8}{\sqrt{3}}\mu+4-b^2>0\); this
concave parabola is positive at zero and its positive root is
\(\frac{2+\sqrt{16-3b^2}}{2\sqrt{3}}\). This quantity is greater than
\(\sqrt{3(1-b)}\) so \(J_1''\) is always positive for
\(\mu \in [0, \sqrt{3(1-b)}[\).

	The proof that the function \(\mu^*(x)\) is injective corresponds to the
proof of the uniqueness of \(x^*(b,\mu)\) in the set \(]x^*_+(b),1[\).
In order to obtain the uniqueness in the whole
\(\bigl]\frac{1}{2},1\bigr[\), we prove now that \(\tilde \phi\) is
strictly decreasing in \(\bigl]\frac{1}{2},x^*_+(b)\bigr[\), so that any
critical point \(x^*\) cannot live in this set. For the continuity of
\(\mu^*\) and since \(\mu^*(1)=-\infty\), for every \(\mu\leq 0\) there
exists \(x>x^*_+(b)\) such that \(\mu^*(x)=\mu\). Furthermore, for any
fixed \(x\), we have proven that there is at most one possible
\(\mu\leq 0\) such that \(\tilde \phi'(x;b,\mu)=0\). So if
\(x< x^*_+(b)\), there is no \(\mu\leq 0\) satisfying the latter
equation.

Consider now \(\mu> 0\). Firstly, we prove that \(\tilde \phi'\) is
increasing with respect to \(\mu\). From \cref{Derfeq} it is evident
that it holds true iff
\(\mu<\frac{(1-x)(2x^2-8x-1)}{\sqrt{1-x^2}(2x^2-2x-1)}\). The right hand
side is a decreasing function for \(x\in\bigl]\frac{1}2,1\bigr[\) so it
is enough to check the inequality in \(x^*_+(b)\). Here the condition
becomes \(\mu<\frac{b^2-8}{b^2+4b-8}\sqrt{3(1-b)}\) and this is greater
than \(\sqrt{3(1-b)}\), so the inequality holds true.

So in order to prove that \(\tilde \phi'(x;b,\mu)\) is negative, it is
enough to prove that \(\tilde \phi'(x;b,\sqrt{3(1-b)})\) is
non-positive. Suppose now that there exists \(x< x^*_+(b)\) such that
\(\tilde \phi'(x;b,\sqrt{3(1-b)})> 0\). Since
\(\tilde \phi'(\frac{1}2;b,\sqrt{3(1-b)})\) is negative, there is an
intermediate point at which \(\tilde \phi'\) is null and then becomes
positive up to \(x\). Also \(\tilde \phi'(x^*_+(b);b,\sqrt{3(1-b)})\) is
null so there is \(x^*_-< x^*_+(b)\) such that
\(\tilde \phi'(x^*_-;b,\sqrt{3(1-b)})>0\) and
\(\tilde \phi''(x^*_-;b,\sqrt{3(1-b)})= 0\). However it holds that
\(\tilde \phi''(x^*_-;b,\sqrt{3(1-b)}) = \frac{J_1(x^*_-)j_2''(x^*_-)-J_1''(x^*_-)j_2(x^*_-)}{j_2(x^*_-)^2} - 2\tilde \phi'(x^*_-)\frac{j_2'(x^*_-)}{j_2(x^*_-)}\).
We have already proven that \(J_1''\) is positive for all
\(x\in\bigl]\frac{1}{2},1\bigr[\) so the first term is positive in
\(x^*_-< x^*_+(b)\). Then, it must be
\(\tilde \phi'(x^*_-)\frac{j_2'(x^*_-)}{j_2(x^*_-)}>0\) or equivalently
\(j_2'(x^*_-)< 0\). Note that \(J_1'(x^*_-)\) is negative because
\(J_1''\) is positive in \(\bigl]\frac{1}2,1\bigr[\) and for
\(\mu=\sqrt{3(1-b)}\) the function \(J_1\) is null at \(x^*_+(b)\). Then
\(\tilde \phi'(x^*_-) = \frac{J_1(x^*_-)j_2'(x^*_-)-J_1'(x^*_-)j_2(x^*_-)}{j_2(x^*_-)^2}\)
is negative, which is a contradiction.

	To sum up, we have shown that for fixed \(b\in[0,1[\) and
\(\mu<\sqrt{3(1-b)}\), there is only one \(x\) such that
\(\tilde \phi'(x;b,\mu)=0\) and this \(x\) lives in \(]x^*_+(b),1[\).
Furthermore, for fixed \(b\in[0,1[\) and \(x\in]x^*_+(b),1[\), there is
only one \(\mu\) in \(]-\infty,\sqrt{3(1-b)}[\) such that
\(\tilde \phi'(x;b,\mu)=0\). If \(b=1\) and \(\mu< 0\), the function
\(\tilde \phi(x;1,\mu)\) is decreasing and reaches its infimum
\(-\frac{\mu}{2}\) at \(x=1\).

\subsubsection{Numerical illustration of
\Cref{PropVanishingUp}}\label{numerical-illustration-of}

	\Cref{FigureVanishingUpwardProp} compares the no arbitrage sub-domain for \(\sigma\)
reported in \Cref{LemmaVanishingSubDomain} with the full no arbitrage
domain found in \Cref{PropVanishingUp}, as functions of \(b\). The red
line represents the sub-domain. The blue line corresponds to the full
domain for \(x^*\) equal to the zero of the \(\mu^*\)-function. The
reason why we take this point, which depends on \(b\), is that the
sub-domain is defined for \(\mu\leq 0\), so \(x^*\) must be greater than
the zero of \(\mu^*(\cdot)\). Then, the green line shows the full domain
for \(x^*=0.99\) (for the chosen values of \(b\), \(0.99\) is always
greater than the zero of \(\mu^*(\cdot)\)) and the light blue line is
the full domain for an intermediate point between the blue and the green
ones. As expected, the red line is above all the other lines.

\begin{figure}
	\centering
	\includegraphics[width=.7\textwidth]{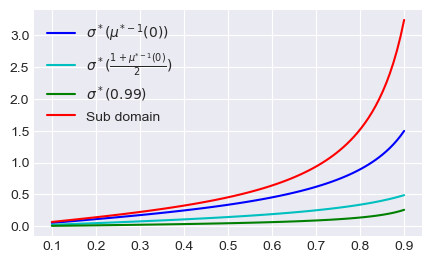}
	\caption{Comparison between the Vanishing Upward SVI arbitrage-free sub-domain and domain as functions of $b$.}
	\label{FigureVanishingUpwardProp}
\end{figure}
    
\subsection{Vanishing (Downward) SVI}\label{vanishing-downward-svi}

	The Vanishing Downward SVI is the sub-SVI obtained by setting
\(\rho=-1\) and \(a=0\). The corresponding SVI formula becomes
\begin{equation*}
SVI(k;0,b,-1,m,\sigma) = b(-k+m + \sqrt{(k-m)^2+\sigma^2}).
\end{equation*}

	We plot in \Cref{FigureVanishingDownward} a Vanishing Downward SVI with \(b=\frac{1}{2}\), \(m=1\)
and \(\sigma=1\).

\begin{figure}
	\centering
	\includegraphics[width=.7\textwidth]{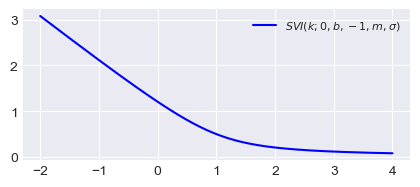}
	\caption{Vanishing Downward SVI with $b=\frac{1}{2}$, $m=1$ and $\sigma=1$.}
	\label{FigureVanishingDownward}
\end{figure}
    
	From
\Cref{PrepInverseSVI}, we know that the Vanishing Downward SVI is
arbitrage-free iff \(SVI(k;0,b,1,-m,\sigma)\) is arbitrage-free, and
this corresponds to a Vanishing Upward SVI. This means that the previous
results still hold for the Vanishing Downward SVI setting
\(m(\rho=-1)=-m(\rho=1)\). In particular, we redefine the quantity
\(\sigma^*\) as
\[\sigma^*(x) := -\frac{4b\sqrt{1-x^{2}}(1-x -2x^{2})}{4\bigl(2-x+\mu^*(x) \sqrt{1-x^{2}}\bigr)^2-b^2(1+x)^2}\]
and \Cref{PropVanishingUp} becomes:

\begin{proposition}[Fully explicit no arbitrage domain for the Vanishing Downward SVI]\label{PropVanishingDown}
A Vanishing Downward SVI with $b=1$ is arbitrage-free iff $\mu> 0$ and $\sigma\geq\frac{\mu}{2}$.

A Vanishing Downward SVI with $0< b< 1$ is arbitrage-free iff it can be parametrized as
\begin{equation}\label{eqSVINewParamVanishingDown}
SVI(k) = b\sigma\Biggl(-\frac{k}\sigma-\mu^*(x) + \sqrt{\biggl(\frac{k}\sigma+\mu^*(x)\biggr)^2+1}\Biggr)
\end{equation}
where $\frac{2+b}{4-b}< x< 1$ and $\sigma\geq \sigma^*(x)$.
\end{proposition}

\section{Extremal Decorrelated SVI}\label{extremal-decorrelated-svi}
	
	The coefficient \(\rho\) in SVI should correspond to the leverage factor
in stochastic volatility models, i.e.~the stock/vol returns correlation.
In particular, when it is zero, the volatility is independent from the
stock which leads to symmetric smiles. In terms of SVI parameters this
means that \(m\) should also be zero.

This is not automatically enforced in SVI, where \(\rho\) and \(m\) are
distinct parameters. Therefore we call:

\begin{itemize}
	\item
	\emph{Decorrelated SVI} sets of SVI parameters where \(\rho=0\);
	\item
	\emph{Extremal Decorrelated SVI} a Decorrelated SVI with \(b=2\);
	\item
	\emph{Symmetric SVI} a Decorrelated SVI with \(m=0\).
\end{itemize}

The two latter families intersect into an SVI with \(b=2\) and
\(\rho=m=0\).

	The \emph{Extremal Decorrelated} SVI is the sub-SVI obtained by setting
\(b=2\) and \(\rho=0\). The corresponding SVI formula becomes
\begin{equation*}
SVI(k;a,2,0,m,\sigma) = a+2\sqrt{(k-m)^2+\sigma^2}.
\end{equation*}

With our notations, \(N(l) = \gamma + \sqrt{l^2+1}\).

	We plot in \Cref{FigureExtremal} an Extremal Decorrelated SVI with \(a=8\), \(m=2\) and
\(\sigma=2\).

\begin{figure}
	\centering
	\includegraphics[width=.7\textwidth]{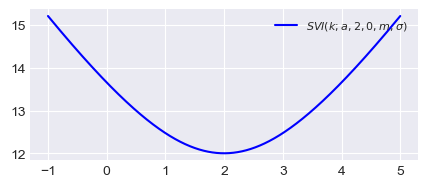}
	\caption{Extremal Decorrelated SVI with $a=8$, $m=2$ and $\sigma=2$.}
	\label{FigureExtremal}
\end{figure}
    
\subsection{The Fukasawa conditions}\label{the-fukasawa-conditions}

	The Roger Lee conditions on \(b\) and \(\rho\) are satisfied. In
subsection 5.2.5 of \cite{martini2020no} we prove that the Fukasawa
conditions are satisfied iff \(\gamma>\tilde F(2,0)=0\) and
\(\mu\in I_{\gamma,2,0}=]-\gamma,\gamma[\). Setting \(\mu=q\gamma\) we
obtain the following.

	\begin{lemma}[Fukasawa conditions for the Extremal Decorrelated SVI]

An Extremal Decorrelated SVI satisfies the Fukasawa conditions iff it can be parametrized by
\begin{equation}\label{eqSVINewParamExtremalFuk}
2 \sigma \biggl(\gamma + \sqrt{\Bigl(\frac{k}{\sigma}-q\gamma\Bigr)^2+1} \biggr)
\end{equation}
with $\gamma >0$, $q \in ]-1,1[$ and $\sigma>0$.

\end{lemma}

\subsection{The condition on $\sigma$}\label{the-condition-on-sigma}

	We will prove the following characterization of the no Butterfly
arbitrage domain for the Extremal Decorrelated SVI:

\begin{proposition}[Fully explicit no arbitrage domain for the Extremal Decorrelated SVI]\label{PropExtremal}

An Extremal Decorrelated SVI is arbitrage-free iff it can be parametrized as
\begin{equation}\label{eqSVINewParamExtremal}
SVI(k) = 2 \sigma \biggl(\gamma + \sqrt{\Bigl(\frac{k}{\sigma}-q\gamma\Bigr)^2+1} \biggr)
\end{equation}
with $\gamma >0$, $q \in ]-1,1[$ and $\sigma\geq\frac{1}{\gamma(1-|q|)}$.

\end{proposition}

\subsubsection{Proof of
\Cref{PropExtremal}}\label{proof-of}

	Setting \(\mu=q\gamma\) with \(q\in]-1,1[\), the function \(G_1\) is
\begin{equation*}
G_1(l) = \frac{- 2q\gamma l^3 + l^2((3 + q^2)\gamma^{2} + 3) - 4q\gamma l + 4(1 + \gamma^2) + 2\gamma\sqrt{l^{2} + 1}(l^2 - 2q\gamma l + 4)}{4(l^2+1)(\gamma+\sqrt{l^2+1})^2}
\end{equation*} and since \(\rho=0\),
\(g_2(l) = -\frac{(l^2-2)\sqrt{l^2+1}-2\gamma}{2(l^2+1)^\frac{3}{2}(\gamma +\sqrt{l^2+1})}\)
is symmetric, so \(l_1=-l_2\). In particular, working with \(l\geq 0\),
to find its roots it is sufficient to solve a third degree equation in
\(\sqrt{l^2+1}\).

	We operate the substitution \(y=\sqrt{l^2+1}\) so that \begin{align*}
G_1(l(y)) &= \frac{2\gamma y^3 + ((3+q^2)\gamma^2+3)y^2+ 6\gamma y + 1 + (1-q^2)\gamma^2 -2q\gamma l(y)(y^2+2\gamma y+1)}{4y^2(y+\gamma)^2},\\
g_2(l(y)) &= -\frac{y^3-3y-2\gamma}{2y^3(y+\gamma)}.
\end{align*}

	From \Cref{Lemmafrho0}, it follows that the function
\(f=-\frac{bg_2}{2G_1}\) is such that for \(l>l_2\) and \(q\) negative,
\(f(l)< f(-l)\) while for \(q\) positive, \(f(l)> f(-l)\) and for
\(q=0\) it is symmetric. So if we want to find the supremum of \(f\)
in \(]-\infty,l_1[\cup]l_2,\infty[\), it is enough to look at
\(]-\infty,l_1[\) for \(q< 0\) and \(]l_2,\infty[\) for \(q\geq0\). Note
also that \(f(l;\gamma,q)=f(-l;\gamma,-q)\).

	Suppose \(q\geq 0\) and \(l>l_2\). The function \(f\) has its limit at
\(\infty\) equal to \(\frac{1}{\gamma(1-q)}\). We want to prove that
this is also the supremum of \(f\). This would be true iff
\(f(l)<\frac{1}{\gamma(1-q)}\) for every \(l\) or equivalently iff
\(-g_2(l)\gamma(1-q)< G_1(l)\). Using the change of variable, this reads
also
\[\frac{y^3-3y-2\gamma}{y^3(y+\gamma)}\gamma(1-q) < \frac{2\gamma y^3 + ((3+q^2)\gamma^2+3)y^2+ 6\gamma y + 1 + (1-q^2)\gamma^2 -2q\gamma\sqrt{y^2-1}(y^2+2\gamma y+1)}{2y^2(y+\gamma)^2}.\]
Multiplying by \(2y^3(y+\gamma)^2\) and simplifying, one gets
\begin{multline*}
2 \gamma q y^{4} + (\gamma^{2}(1+q)^{2} + 3)y^{3} + 6\gamma(2-q)y^{2} + (\gamma^{2} (1-q)(q+11) + 1)y +\\
4 \gamma^{3} (1-q) - 2\gamma qy\sqrt{y^2-1}(y^{2} + 2\gamma y + 1) >0.
\end{multline*} Since \(\sqrt{y^2-1}< y\), the LHS is greater than
\[(\gamma^{2} (1-q)^{2} + 3)y^{3} + 4\gamma(3-2q)y^{2} + (\gamma^{2} (1-q)(q+11) + 1)y + 4 \gamma^{3} (1-q).\]
Each of the coefficients of this polynomial is positive and \(y>0\) so
the whole polynomial is positive and \(f\) reaches its supremum at
\(\infty\).

	In the case \(q<0\), the conclusion follows immediately from
\(\Cref{PrepInverseSVI}\).

\section{Symmetric SVI}\label{symmetric-svi}

	The Symmetric SVI is the sub-SVI obtained by setting \(\rho=0\) and
\(m=0\). The corresponding SVI formula becomes \begin{equation*}
SVI(k;a,b,0,0,\sigma) = a+b\sqrt{k^2+\sigma^2}.
\end{equation*} With our notations, \(N(l) = \gamma+\sqrt{l^2+1}\).

	We plot in \Cref{FigureSymmetric} a Symmetric SVI with \(a=0.64\), \(b=1.6\) and
\(\sigma=0.4\).

\begin{figure}
	\centering
	\includegraphics[width=.7\textwidth]{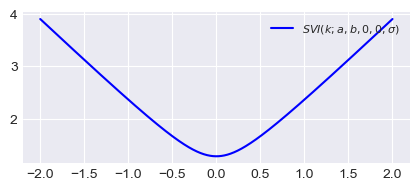}
	\caption{Symmetric SVI with $a=0.64$, $b=1.6$ and $\sigma=0.4$.}
	\label{FigureSymmetric}
\end{figure}
    
\subsection{The Fukasawa conditions}\label{the-fukasawa-conditions}

	The Roger Lee conditions imply \(0< b\leq 2\). In Appendix B of
\cite{martini2020no} we prove that the Fukasawa conditions are satisfied
iff
\(\gamma>\tilde F(b,0)=g_{-(b,0)}\Bigl(-\frac{6b}{\sqrt{b^4-20b^2+64}}\Bigr)=-\frac{(b^2+32)\sqrt{4-b^2}}{(16-b^2)^{\frac{3}{2}}}\).
Since \(\rho=0\), the interval for \(\mu\) is symmetric and being
\(\mu=0\), the Fukasawa condition on \(\mu\) is automatically satisfied.

	\begin{lemma}[Fukasawa conditions for the Symmetric SVI]

A Symmetric SVI with $0< b\leq 2$ satisfies the Fukasawa conditions iff $\gamma>\tilde F(b)=-\frac{(b^2+32)\sqrt{4-b^2}}{(16-b^2)^{\frac{3}{2}}}$.

\end{lemma}

	For every value of \(b\), we have \(\tilde F(b)\in ]-1,0]\) and in
particular \(\tilde F(2)=0\). We can extend the definition of
\(\tilde F\) to the point \(b=0\), setting \(\tilde F(0)=-1\).
Furthermore,
\(\frac{d\tilde F}{db}(b) = \frac{108b^3}{(16-b^2)^{\frac{5}2}\sqrt{4-b^2}}>0\)
for positive \(b\). So the inverse function
\(\tilde F^{-1}:[-1,0]\to[0,2]\) is well defined. As in
\cref{the-b-approach-reparametrizing-from-the-critical-point-equation},
we define \begin{equation}\label{eqGtilde1}
\tilde G(\gamma) =
\begin{cases}
\tilde F^{-1}(\gamma) & \text{if}\ \gamma\in]-1,0],\\
2 & \text{if}\ \gamma>0.
\end{cases}
\end{equation} It can be shown, working out the calculation of the roots
of a 3rd degree polynomial, that the explicit formula of \(\tilde G\) in
the case \(\gamma\leq 0\) is \begin{equation}\label{eqGtilde2}
\tilde G(\gamma) = 2\sqrt{\frac{6\sqrt{8\gamma^2+1}\cos\biggl(\frac{1}{3}\arccos\biggl(-\frac{8\gamma^4+20\gamma^2-1}{(8\gamma^2+1)^{\frac{3}2}}\biggr)\biggr) -4\gamma^2-5}{1-\gamma^2}}.
\end{equation}

\subsection{The condition on $\sigma$}\label{the-condition-on-sigma}

	Before enunciating the main Proposition of this chapter, we need to
introduce some notation. In this section, for positive \(l\)s we will
operate a change of variable setting \(z=\frac{1}{\sqrt{l^2+1}}\), whose
domain is \(]0,1]\). Then we define the functions \(J_1(z):=G_1(l(z))\),
\(j_2(z):=g_2(l(z))\), \(\eta(z):=h(l(z))\) and \(j(z):=g(l(z))\). In
particular, \begin{equation}\label{eqj2hg}
\begin{aligned}
j_2(z;\gamma) &= z\frac{2\gamma z^3+3z^2-1}{2(\gamma z+1)},\\
\eta(z;\gamma) &= 1-\frac{1-z^2}{2(1+\gamma z)},\\
j(z) &= \frac{1-z^2}{4}.
\end{aligned}
\end{equation}

In the following, the symbol \('\) indicates the derivative with respect
to \(z\) if not differently specified.

We also define the functions \begin{align}
\gamma^*(u) &:= u\sqrt{6u^3+15u^2+14u+6+(1+u)^2\sqrt{3(12u^2+12u+11)}},\label{eqgammab0}\\
b^*(z,\gamma) &:= \sqrt{\frac{\eta(z)(\eta(z)j_2'(z) - 2\eta'(z)j_2(z))}{j(z)(j(z)j_2'(z)-2j'(z)j_2(z))}},\label{Symmetriceqbstar}\\
\sigma^*(z,\gamma) &:= -\frac{b^*(z^*,\gamma)j_2(z^*)}{2(\eta(z^*)^2-b^*(z^*,\gamma)^2j(z^*)^2)}.\label{Symmetriceqsigmastar}
\end{align}

	The final result of this subsection will be:

\begin{proposition}[Fully explicit no arbitrage domain for the Symmetric SVI]\label{PropSymmetric}
An arbitrage-free Symmetric SVI must have $\gamma>-1$. 

A Symmetric SVI with $b=2$ is arbitrage-free iff $\gamma>0$ and $\sigma\geq\frac{1}{\gamma}$.

Call $z^*(\gamma^*(u),0)=\frac{u}{\gamma^*(u)}$ and $z^*(\gamma,\tilde G(\gamma)) = \frac{\sqrt{(4-\tilde G(\gamma)^2)(16-\tilde G(\gamma)^2)}}{\tilde G(\gamma)^2+8}$.
\begin{itemize}
	\renewcommand{\labelitemii}{\tiny$\square$}
	\item If $\gamma = -\sqrt{\frac{9+5\sqrt{3}}{18}}$, the Symmetric SVI is arbitrage-free iff $b< 2\sqrt{3\sqrt{3}-5}$ and \newline$\sigma\geq-\frac{bj_2(\hat z)}{2(\eta(\hat z)^2-b^2j(\hat z)^2)}$ where $\hat z=\sqrt{\frac{3-\sqrt{3}}2}$.
	\item If $\gamma\neq-\sqrt{\frac{9+5\sqrt{3}}{18}}$ and $b< 2$, a Symmetric SVI is arbitrage-free iff it can be parametrized as
	\begin{equation}\label{eqSVINewParamSymmetric}
		SVI(k) = \sigma b^*(z,\gamma^*(u))\biggl(\gamma^*(u)+\sqrt{\Bigl(\frac{k}{\sigma}\Bigr)^2+1}\biggr)
	\end{equation}
	where
	\begin{itemize}
		\item $u>-1$,
		\item $z\in]z^*(\gamma^*(u),0),z^*(\gamma^*(u),\tilde G(\gamma^*(u)))[$ if $\gamma^*(u)< -\sqrt{\frac{9+5\sqrt{3}}{18}}$ or \newline$z\in]z^*(\gamma^*(u),\tilde G(\gamma^*(u))),z^*(\gamma^*(u),0)[$ if $\gamma^*(u)>-\sqrt{\frac{9+5\sqrt{3}}{18}}$,
		\item $\sigma\geq\sigma^*(z,\gamma^*(u))$.
	\end{itemize}
\end{itemize}

\end{proposition}

\subsubsection{Proof of
\Cref{PropSymmetric}}\label{proof-of}

	When \(b=2\), we are in the case of the Extremal Decorrelated SVI with
\(q=0\), so the arbitrage conditions are satisfied iff \(\gamma>0\) and
\(\sigma\geq\frac{1}{\gamma}\).

	From now on we consider \(b< 2\). Since \(\rho=0\), the \(g_2\) function
is symmetric with respect to \(l\) and since \(\mu=0\), the \(G_1\)
function is symmetric too. Then \(\tilde f = -\frac{G_1}{g_2}\) is
symmetric as a function of \(l\) and we consider \(l\geq 0\). Recall
that we are interested in the value of
\(\sigma^*=\sup_{l>l_2}\frac{b}{2\tilde f(l)}\).

	The function \(\tilde \phi(z):=\tilde f(l(z))\) goes to \(\infty\) at
\(0\) and at the non trivial zero of \(j_2\), which we call
\(z_2(\gamma)\). Then \(\tilde\phi\) has at least one point of minimum
in \(]0,z_2[\).

\paragraph{Study of $j_2$}\label{study-of-j_2}

	To find \(z_2\), we should solve \(p_1(z):=2\gamma z^3+3z^2-1=0\). It
can be shown that for \(\gamma\) negative, the polynomial has three
roots but only the second one lies in the interval \(]0,1[\). If
\(\gamma\) is null, then \(z_2(0)=\frac{1}{\sqrt{3}}\). If \(\gamma\) is
positive but not greater than \(1\), the polynomial has three roots and
only the third lies in the interval \(]0,1[\). Finally, if \(\gamma>1\),
there is only one root. Using the trigonometric notation,
\(z_2(\gamma)\) is \begin{equation*}
z_2(\gamma) =
\begin{cases}
-\frac{1}{\gamma}\cos\biggl(\frac{1}{3}\arccos(1-2\gamma^2)-\frac{2\pi}3\biggr)-\frac{1}{2\gamma} & \text{if $\gamma< 0$},\\
\frac{1}{\sqrt{3}} & \text{if $\gamma=0$},\\
\frac{1}{\gamma}\cos\biggl(\frac{1}{3}\arccos(2\gamma^2-1)\biggr)-\frac{1}{2\gamma} & \text{if $0<\gamma\leq 1$},\\
\frac{1}{\gamma}\cosh\biggl(\frac{1}{3}\operatorname{arccosh}(2\gamma^2-1)\biggr)-\frac{1}{2\gamma} & \text{if $\gamma>1$}.
\end{cases}
\end{equation*}

	Let us show that \(j_2\) has a single critical point \(z_{m_2}\) where
it achieves its minimum. It holds \begin{align*}
j_2'(z) &= \frac{6\gamma^2z^4+14\gamma z^3+9z^2-1}{2(\gamma z+ 1)^2},\\
j_2''(z) &= \frac{6\gamma^{3} z^{4} + 19 \gamma^{2} z^{3} + 21 \gamma z^{2} + 9 z + \gamma}{(\gamma z+1)^3}.
\end{align*} Since \(j_2'(0)=-\frac{1}2\) and
\(j_2'(1)=\frac{3\gamma+4}{\gamma+1}\), which is positive having
\(\gamma>-1\), then \(j_2\) has at least one critical point in
\(]0,1[\).

\begin{itemize}
\item
  The function \(j_2''\) is positive if \(\gamma\geq 0\) so \(j_2\) has
  exactly one critical point in this case, which is a point of minimum.
\item
  Consider the case \(\gamma< 0\). The polynomial
  \(p_2(z):=6\gamma^2z^4+14\gamma z^3+9z^2-1\) has derivative
  \(p_2'(z)=2z(12\gamma^2z^2+21\gamma z+9)\), so its three critical
  points are \(0\), \(-\frac{3}{4\gamma}\) and \(-\frac{1}{\gamma}\).
  Since \(\gamma>-1\), the third critical point, which is a point of
  minimum, is greater than \(1\) and, since \(0\) is the first critical
  point and since
  \(-1=p_2(0)< 0< p_2(1)=6\gamma^2+14\gamma +8 = 2(\gamma+1)(\gamma+4)\),
  \(p_2\) has exactly one zero in \(]0,1[\). So again \(j_2\) has only
  one critical point in this interval.
\end{itemize}

Because \(j_2\) is negative before \(z_2\) then \(j_2'(z_2)>0\),
furthermore \(j_2'(0)< 0\) so the critical point \(z_{m_2}\) of \(j_2\)
(a point of minimum) lies in \(]0,z_2[\).

	We have already seen that if \(\gamma\geq0\), then \(j_2\) is strictly
convex in \(]0,1[\).

We can moreover show that \(j_2\) has a single inflection point if
\(\gamma\) is negative. In such case, call the numerator of \(j_2''\) as
\(p_3(z):=6\gamma^{3} z^{4} + 19 \gamma^{2} z^{3} + 21 \gamma z^{2} + 9 z + \gamma\).
It holds \(p_3(0)=\gamma< 0\) and
\(p_3(1)=6\gamma^3+19\gamma^2+22\gamma+9\), which is a third degree
polynomial in \(\gamma\) with only root equal to \(-1\) and a positive
leading order term, so it is positive for \(\gamma>-1\). Then \(p_3\)
has at least one zero in \(]0,1[\). Moreover,
\(p_3'(z)=3(\gamma z+1)^2(8\gamma z+3)\), whose zeros are
\(-\frac{3}{8\gamma}\) and \(-\frac{1}{\gamma}>1\). The former is a
point of maximum for \(p_3\) while the latter is an inflection point.
Gathering all the informations, it follows that \(p_3\) has exactly one
zero in \(]0,1[\). Since \(j_2''>0\) at the point of minimum of \(j_2\)
and \(j_2''(0)< 0\), the only inflection point \(z_{i_2}(\gamma)\) of
\(j_2\) lies in \(]0,z_{m_2}[\).

\paragraph{Study of $J_1$}\label{study-of-j_1}

	It holds
\(J_1(z) = \eta(z)^2-b^2j(z)^2 = \Bigl(1-\frac{1-z^2}{2(1+\gamma z)}\Bigr)^2-b^2\frac{1-z^2}{16}\).

	We write here a general result which holds for every \(SVI\).

\begin{lemma}

Under the Fukasawa conditions, $G_1$ is strictly decreasing at the second zero of $g_2$.

\end{lemma}

\begin{proof}

Since $G_1' = G_{1+}'G_{1-}+G_{1+}G_{1-}'$ and $G_{1+}'< G_{1-}'$ for $l>l^*$, then it is enough to prove $G_{1-}'(l_2)< 0$. Remember that $G_{1-}(l_2) = 1-N'(l_2)\bigl(\frac{l_2+\mu}{2N}-\frac{b}4\bigr)$. From the formula of $g_2$, we have that $N''(l_2)=\frac{N'(l_2)^2}{2N(l_2)}$. Then
\begin{align*}
G_{1-}'(l_2) &= -\frac{N''(l_2)(l_2+\mu)+N'(l_2)}{2N(l_2)} + \frac{N'(l_2)^2(l_2+\mu)}{2N(l_2)^2} + \frac{bN''(l_2)}{4}\\
&= \frac{N'(l_2)^2(l_2+\mu)}{4N(l_2)^2} + \frac{bN'(l_2)^2}{8N(l_2)} - \frac{N'(l_2)}{2N(l_2)}.
\end{align*}
For the absence of arbitrage, $\mu<\inf_{l>l^*}L_+(l)=\inf_{l>l^*}\Bigl(\frac{2N(l)}{N'(l)}-\frac{bN(l)}{2}-l\Bigr)$ so this holds in particular in $l_2>l^*$. Substituting with the formula of $L_+$ we obtain
$$G_{1-}'(l_2)< \frac{N'(l_2)^2}{4N(l_2)^2}\biggl(l_2+\frac{2N(l_2)}{N'(l_2)}-\frac{bN(l_2)}{2}-l_2\biggr) + \frac{bN'(l_2)^2}{8N(l_2)} - \frac{N'(l_2)}{2N(l_2)} = 0.$$
\end{proof}

	Observe that \begin{align*}
J_1'(z) &= \frac{\gamma z^4(b^2\gamma^2+4) + z^3(3b^2\gamma^2+8\gamma^2+8) + 3\gamma z^2(b^2+8) + z(b^2+8\gamma^2+8) + 4\gamma}{8(\gamma z+1)^3},\\
J_1''(z) &= \frac{b^2}8 + \frac{\gamma^2z^4+4\gamma z^3+6z^2+(4\gamma z+1)(2-\gamma^2)}{2(\gamma z+1)^4},\\
J_1'''(z) &= \frac{6z(\gamma^2-1)^2}{(\gamma z+1)^5}.
\end{align*} So:

\begin{itemize}
\item
  For \(\gamma\geq 0\), \(J_1'\) is always greater than \(0\) so \(J_1\)
  is increasing and has its minimum in \(0\).
\item
  If \(\gamma< 0\), \(J_1'\) is negative for \(z=0\) and we know that
  \(G_1'(l_2)< 0\) so \(J_1'(z_2)>0\). Then \(J_1\) has at least a point
  of minimum in \(]0,z_2[\). Also, the minimum of \(J_1''\) is reached
  at \(0\) and it is \(\frac{b^2-4\gamma^2+8}8\). This quantity is
  positive iff \(\gamma^2<\frac{b^2+8}4\), so for example for every
  \(\gamma<0\) because \(\gamma\) is greater than \(-1\). In such case
  \(J_1\) is always convex and it has exactly one critical point
  \(z_{m_1}\), which is a point of minimum, in \(]0,z_2[\).
\end{itemize}

	Note also that when \(\gamma\) goes to its Fukasawa limit,
\(J_{1+}(z)= 1-\frac{1-z^2}{2(\gamma z+1)}-\frac{b\sqrt{1-z^2}}{4}\)
goes to \(0\) in \(z^*_+(b):=\frac{\sqrt{(4-b^2)(16-b^2)}}{b^2+8}\),
which is also a point of minimum for \(J_1\). It is easy to shown that
if
\(\gamma>-\frac{(b^4-38b^2+64)(b^2+8)}{(4-b^2)^{\frac{3}2}(16-b^2)^{\frac{3}2}}:=M(b)\),
then \(j_2(z_+^*(b)))>0\), so \(z_+^*(b)>z_2(\gamma)\), otherwise
\(z_+^*(b)\leq z_2(\gamma)\). Note that \(M(b)\geq\tilde F(b)\) (the
equality holding iff \(b=0\)), so when \(\gamma\) goes to
\(\tilde{F}(b)\), \(J_1\) reaches its minimum before \(z_2(\gamma)\) (or
at \(z_2(\gamma)\) if \(b=0\)).

Then \(J_1\) has in any case one point of minimum: in the case
\(\gamma\geq 0\), this is \(0\), while when \(\gamma<0\), it is
\(z_{m_1}(\gamma,b)\). Also, \(J_1\) is convex if
\(\gamma^2\leq\frac{b^2+8}4\) (in particular \(\gamma<0\) given that
\(\gamma\) should be larger than \(-1\)) while it has a unique
inflection point \(z_{i_1}(\gamma,b)\) if \(\gamma^2>\frac{b^2+8}4\).

\paragraph{Uniqueness of the critical point of $\tilde\phi$ for the Symmetric SVI}\label{uniqueness-of-the-critical-point-of-tildephi-for-the-symmetric-svi}

	From the formula for \(\tilde \phi' = \frac{J_1j_2'-J_1'j_2}{j_2^2}\),
given that \(J_1 > 0\) and \(j_2 < 0\) we have that:

\begin{itemize}
\item
  for \(\gamma< 0\): all the critical points of \(\tilde \phi\) must lie
  in \([z_{m_1},z_{m_2}] \cup [z_{m_2},z_2\land z_{m_1})\);
\item
  for \(\gamma\geq0\): all the critical points of \(\tilde \phi\) must
  lie in \(]0,z_{m_2}[\) since \(J_1' j_2\) is negative.
\end{itemize}

Considering now the second derivative of \(\tilde \phi\), so
\(\tilde \phi'' = \frac{J_1j_2''-J_1''j_2}{j_2^2} - 2\tilde \phi'\frac{j_2'}{j_2}\)
and evaluating it at a point of maximum of \(\tilde \phi\), the second
term becomes null and necessarily \(\tilde \phi'' \leq 0\). Then it must
hold \(j_2''J_1''\leq0\) (the equality holding when both factors are
\(0\)). In particular:

\begin{enumerate}
\def\labelenumi{\arabic{enumi}.}
\item
  for \(\gamma< 0\): if \(\tilde \phi\) has a point of maximum, it lies
  in \([z_{m_1},z_{i_2}(\gamma)[\) because in such case \(J_1''>0\),
  \(j_2''<0\) in \([0,z_{i_2}(\gamma)[\);
\item
  for \(0\leq\gamma\leq\sqrt{\frac{b^2+8}4}\): \(\tilde \phi\) cannot
  have a point of maximum, because \(J_1\) and \(j_2\) are strictly
  convex, so it has one point of minimum which lies in \(]0,z_{m_2}[\);
\item
  for \(\gamma>\sqrt{\frac{b^2+8}4}\): \(J_1\) has one inflection point
  \(z_{i_1}\) and if \(\tilde \phi\) has a point of maximum, it lies in
  \(]0,z_{m_2}\land z_{i_1}[\).
\end{enumerate}

	We consider the first case and show that \(z_{m_1}>z_{i_2}(\gamma)\) or
equivalently that \(J_1'(z_{i_2})< 0\), so that there is no point of
maximum. The quantity \(z_{i_2}(\gamma)\) does not depend on \(b\) so
\(\partial_{b^2}(J_1'(z_{i_2})) = \partial_{b^2}J_1'(z_{i_2}) = \frac{z_{i_2}}{8}>0\).
We can then check that for every \(-1<\gamma< 0\) we have indeed
\(J_1'(z_{i_2};b=2)< 0\) in \Cref{FigureSymmetricProof}.

\begin{figure}
	\centering
	\includegraphics[width=.7\textwidth]{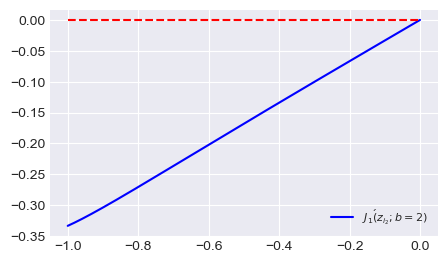}
	\caption{$J_1'(z_{i_2};b=2)$ as a function of $\gamma$.}
	\label{FigureSymmetricProof}
\end{figure}
    
	Let us now consider the last case. Since
\(\gamma>\sqrt{\frac{b^2+8}4}\), then in particular \(\gamma^2>2\).
Suppose we are at a critical point of \(\tilde \phi\), then
\(J_1 = \frac{J_1'}{j_2'}j_2\) and its second derivative evaluated at
this point is
\(\tilde \phi'' = \frac{J_1j_2''-J_1''j_2}{j_2^2} = \frac{J_1'j_2''-J_1''j_2'}{j_2j_2'}\).
In the set \(]0,z_{m_2}\land z_{i_1}[\), the denominator is always
positive. The numerator times \(16(\gamma z + 1 )^6\) is equal to
\begin{align*}
&b^{2}\bigl(6 \gamma^{6} z^{8} + 36 \gamma^{5} z^{7} + 91 \gamma^{4}z^{6} + 126 \gamma^{3} z^{5} + 3\gamma^2(\gamma^{2} + 34)z^{4} + 2\gamma(5\gamma^{2} + 23)z^{3} + 3(4\gamma^{2} + 3)z^{2} + 6 \gamma z + 1\bigr)+\\
&+ 24 \gamma^{4} z^{8} + 96\gamma^3(\gamma^{2} + 1)z^{7} + 4\gamma^2 z^{6}(148\gamma^{2} + 17) + 24\gamma z^{5}(8\gamma^{4}+48\gamma^{2}-3) + 12z^{4}(50\gamma^{4} + 81\gamma^{2} - 6)+\\
&+16\gamma z^{3}(44\gamma^{2} + 25) + 12z^{2}(33 \gamma^{2} + 8) + 120 \gamma z + 4 \gamma^{2} + 8.
\end{align*} All coefficients of \(z\) are positive quantities, indeed
the only minus involved are in the expressions
\(8\gamma^{4}+48\gamma^{2}-3> 8\times4+48\times2-3>0\) and
\(50\gamma^{4} + 81\gamma^{2} - 6 > 50\times4 + 81\times2 - 6>0\). Then
\(\tilde \phi''\) is positive and \(\tilde \phi\) cannot have a point of
maximum even in this case.

	We denote from now on the well-defined single point of minimum of
\(\tilde \phi\) with \(z^*(\gamma,b)\).

\paragraph{Limits of $z^*$ with respect to
$\gamma$}\label{limits-of-z-with-respect-to-gamma}

	When \(\gamma\) goes to \(\infty\), the ratio between the numerator of
\(\tilde \phi'\) and \(\gamma^4\) goes to \(2z^6(48+b^2(z^2-3))\), which
is null for \(z=0\), so \(z^*(\infty,b)=0\).

On the other side, let \(\gamma\) go to \(\tilde F(b)\). Then, from the
above study of \(J_1\), the minimum point of \(J_1\) is
\(z^*_+(b)=\frac{\sqrt{(4-b^2)(16-b^2)}}{b^2+8}< z_2(\tilde F(b))\) and
it coincides with the minimum point \(z^*(\tilde F(b),b)\) of
\(\tilde \phi\), at which \(\tilde \phi\) is null.

If \(b=0\), it still holds \(z^*(\infty,0)=0\) and it is easy to show
that at the point \(z^*(0)=z_2(\tilde F(0))=1\) the function
\(\tilde\phi\) with \(\gamma = \tilde F(b)\) vanishes.

\paragraph{The $b^*$ approach for the Symmetric
SVI}\label{the-b-approach-for-the-symmetric-svi}

	Let \(\gamma\in]-1,\infty[\) be fixed. With reference to the definition
of \(\tilde G(\gamma)\) in \cref{eqGtilde1}, we want to prove that for
every \(b\in]0,\tilde G(\gamma)[\) the corresponding \(z^*(\gamma,b)\)
lies in \(]z^*(\gamma,0),z^*(\gamma,\tilde G(\gamma))[\) (or in
\(]z^*(\gamma,\tilde G(\gamma)),z^*(\gamma,0)[\)).

\begin{remark}

Note that if $\gamma>0$, $b$ could actually attain the value $\tilde G(\gamma)=2$ but we have already seen in the Extremal Decorrelated scenario that in such case the function $\tilde f$ is increasing and has its minimum (which is not necessarily a critical point) in $z=0$. So in the following we will always consider $b< \tilde G(\gamma)$.

\end{remark}

	Remember the definition of \(\eta\) and \(j\) in \cref{eqj2hg} and the
fact that we use \('\) to indicate the derivative with respect to \(z\).

	The proof of the uniqueness of the critical point of \(\tilde \phi\)
still holds for \(b=0\) and \(b= \tilde G(\gamma)\), and we proved
\(z^*(\gamma,\tilde G(\gamma))=z^*_+(G(\gamma))\). In particular, the
single point of minimum \(z^*(\gamma,b)\) defined above of
\(\tilde \phi\) lies in \(]0,z_{2}[\).

In a following paragraph we show that the only solution to
\(p(z)=q(z)=0\) is for \(\hat\gamma=-\sqrt{\frac{9+5\sqrt{3}}{18}}\) at
\(\hat z = \sqrt{\frac{3-\sqrt{3}}2}\). Take \(\gamma\neq\hat \gamma\).
With reference to the observations in \cref{the-b-approach-reparametrizing-from-the-critical-point-equation}, we can
conclude that the set \(Z(\gamma)\) of the critical points of
\(\tilde \phi\) is equal to the interval with extrema the only zero of
\(p=\eta j_2' - 2\eta'j_2\), denoted \(z^*(\gamma,0)\), and the point
\(z^*(\gamma,\tilde G(\gamma))\).

	For \(\gamma=\hat\gamma\), the minimum point of \(\tilde \phi\) does not
depend on \(b\) and it is \(z^*(\hat\gamma,b)=\hat z\). In such case,
\(\tilde G(\hat\gamma)=2\sqrt{3\sqrt{3}-5}\). In the following we see
what happens for the other values of \(\gamma\).

	It holds
\(q(z) = -\frac{2 z^4\gamma^{2}(z^2-3) + 4z^3\gamma(z^2-3) + 3 z^{4} - 8 z^{2} + 1}{8 \sqrt{1 - z^{2}}(\gamma z + 1)^2}\).
Suppose \(z_q(\gamma)\) is a zero of \(q\), then solving the equation
\(q(z_q)=0\) in \(\gamma\), the only solution greater than \(-1\) is
\(\gamma = \gamma_q(z_q): = \frac{1-z_q^2}{z_q^2\sqrt{2(3-z_q^2)}}-\frac{1}{z_q}\).
It can be shown that
\(\gamma_q'(z) = \frac{1}{z^2}\Bigl(1-\frac{z^4-3z^2+6}{z\sqrt{2}(3-z^2)^{\frac{3}{2}}}\Bigr) <0\)
for every \(z\in]0,1[\). This means that \(\gamma_q\) is invertible from
\(]0,1[\) to \(]-1,\infty[\) and that for fixed \(\gamma\), \(q\) has a
unique zero. Observe that \(\gamma_q(\hat z)=\hat\gamma\). Since
\(q(0)<0\), on the left of its zero \(z_q(\gamma)\), \(q\) is negative
while on the right it is positive. Also, if \(\gamma<\hat\gamma\), then
\(z_q(\gamma)>\hat z\). We now look at
\(p(z) = \frac{2\gamma^{2} z^6 + 12\gamma^{3} z^{5} + 3z^{4}(10\gamma^{2} -1) + 28\gamma z^{3} + 12 z^{2} - 1}{4(\gamma z + 1)^3}\).
Substituting \(\gamma\) with \(\gamma_q(z)\) in the last expression, we
find
\(p\bigl(z;\gamma=\gamma_q(z)\bigr) = -\frac{z^2\bigl(2z^4-3z^2-3+z\sqrt{2}(3-z^2)^{\frac{3}2}\bigr)}{1-z^2}\).
We call \(r(z) := 2z^4-3z^2-3+z\sqrt{2}(3-z^2)^{\frac{3}2}\), then
\(r(\hat z)=r(1)=0\) and \(r(0)=-3<0\). Furthermore
\(r'(z) = -(4z^2-3)(\sqrt{2(3-z^2)}-2z)\), where the second factor is
positive for \(z<1\), so \(r\) is increasing in
\(\bigl[0,\frac{\sqrt{3}}{2}\bigr[\ni\hat z\) and decreasing in
\(\bigl]\frac{\sqrt{3}}{2},1\bigr[\), with unique zeros at \(\hat z\)
and \(1\). This means that for the previous choice of
\(\gamma<\hat \gamma\), \(p(z_q(\gamma))<0\), so
\(z^*(\gamma,0)>z_q(\gamma)\) because \(z^*(\gamma,0)\) is the unique
zero of \(p\). Remember that \(q\) is positive in
\(]z_q(\gamma),z^*(\gamma,0)[\). Since \(b^*\) defined in \cref{eqbstar}
lives where \(p\) and \(q\) have the same sign and since the set
\(Z(\gamma)\) is an interval, the latter will lie on the right of
\(z^*(\gamma, 0)\), indeed on the immediate left we have seen that \(p\)
is negative while \(q\) is positive. The reasoning is similar for
\(\gamma>\hat\gamma\), with the conclusion that the previous set is on
the left of \(z^*(\gamma, 0)\).

\paragraph{Limits of $z^*$ with respect to
$b$}\label{limits-of-z-with-respect-to-b}

	When \(b\) goes to \(0\) and \(\gamma\) is not \(\hat\gamma\), the value
of \(z^*(\gamma,b)\) is the only solution smaller than \(z_2(\gamma)\)
to \(p(z)=0\), or \begin{equation}\label{eqnumpSymm}
2\gamma^{2} z^6 + 12\gamma^{3} z^{5} + 3z^{4}(10\gamma^{2} -1) + 28\gamma z^{3} + 12 z^{2} - 1=0.
\end{equation} Since the equation is of sixth degree, the solution
cannot be written explicitly, but it can be found by a numeric routine
in the interval \(]0,z_2(\gamma)[\).

Another cunning way to find \(z^*(\gamma,0)\) is to use again the trick
of changing the role between parameters and variables. Firstly, note
that if \(\gamma=0\), the only solution to \cref{eqnumpSymm} is
\(z=\frac{1}{\sqrt{6+\sqrt{33}}}\). If \(\gamma\neq0\), substitute
\(z=\frac{u}{\gamma}\) in the above equation, then we look for \(u\)
solving \begin{equation}\label{eqnumpSymmu}
\gamma^4-2u^2(6u^3+15u^2+14u+6)\gamma^2+u^4(3-2u^2)=0.
\end{equation} The interval where \(u\) lives is such that
\(0<\frac{u}{\gamma}<z_2(\gamma)\), or equivalently
\(\frac{\gamma z_2(\gamma)}{u}>1\). If \(u\) is negative, it must hold
\(0>u>\gamma z_2(\gamma)>\gamma>-1\), so for sure \(u>-1\). Imagine to
fix \(u\), then equation \cref{eqnumpSymmu} is quadratic in \(\gamma^2\)
and has four (eventually complex) solutions \(\pm\sqrt{\Gamma_\pm}\),
with
\(\Gamma_\pm = u^2\bigl(6u^3+15u^2+14u+6\pm(1+u)^2\sqrt{3(12u^2+12u+11)}\bigr)\).
The solution \(\gamma\) must have the same sign of \(u\), since \(z>0\).
If \(u\) is positive, the solution \(\sqrt{\Gamma_-}\) is smaller than
\(u\), so also \(\sqrt{\Gamma_-} z_2(\sqrt{\Gamma_-})\) is smaller than
\(u\). Then, the only possible solution could be \(\sqrt{\Gamma_+}\). We
plot in \Cref{FigureMinusGamma} (left) the graph of \(\frac{\sqrt{\Gamma_+} z_2(\sqrt{\Gamma_+})}{u}\) to
check whether it is greater than \(1\). If \(u\) is negative, we plot in \Cref{FigureMinusGamma} (right)
the graphs of \(\frac{-\sqrt{\Gamma_\pm}z_2(-\sqrt{\Gamma_\pm})}{u}\)
and check again whether these quantities are greater than \(1\).

\begin{figure}
	\centering
	\includegraphics[width=.9\textwidth]{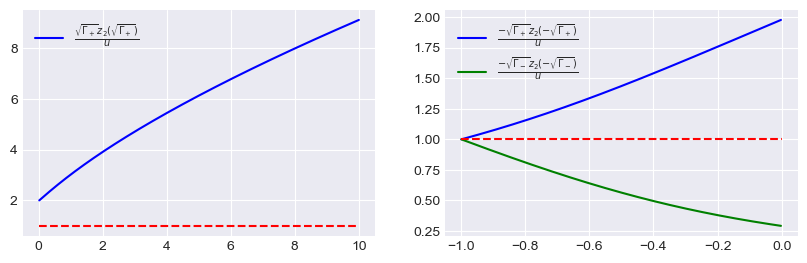}
	\caption{Study of $\frac{\pm\sqrt{\Gamma_\pm}z_2(\pm\sqrt{\Gamma_\pm})}{u}$ as functions of $u$.}
	\label{FigureMinusGamma}
\end{figure}
    
	We can see that for every \(u>0\), the solution \(\sqrt{\Gamma_+}\) is
admissible and that for \(-1<u<0\), the only admissible solution is
\(-\sqrt{\Gamma_+}\). To sum up, for \(u>-1\), the solution to
\cref{eqnumpSymmu} in terms of \(\gamma\) is \cref{eqgammab0} and it
ranges from \(-1\) to \(\infty\). Observe that if \(u=0\), the value of
\(\frac{u}{\gamma^*(u)}\) is exactly \(\frac{1}{\sqrt{6+\sqrt{33}}}\),
so the definition is well posed even in this case.

\subsubsection{Numerical illustration of the interval for $z^*$}\label{numerical-illustration-of-the-interval-for-z}

	We report in \Cref{FigureSymmetriczstar,FigureSymmetriczstar2} the graphs of \(z^*(\gamma,0)\), which is the only zero
of \(p(z,\gamma)\) in \(\bigl]0,z_2(\gamma)\bigr[\), and of
\(z^*(\gamma,\tilde G(\gamma))\).

\begin{figure}
	\centering
	\includegraphics[width=.7\textwidth]{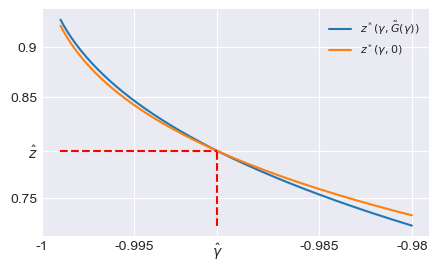}
	\caption{$z^*(\gamma,0)$ and $z^*(\gamma,\tilde G(\gamma))$ as functions of $\gamma$ ranging from $-1$ to $-0.98$.}
	\label{FigureSymmetriczstar}
\end{figure}

\begin{figure}
	\centering
	\includegraphics[width=.7\textwidth]{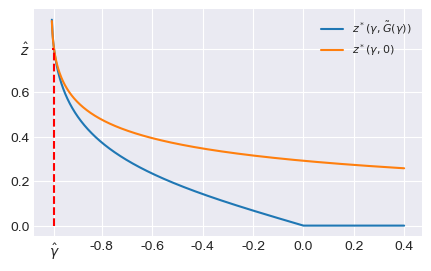}
	\caption{$z^*(\gamma,0)$ and $z^*(\gamma,\tilde G(\gamma))$ as functions of $\gamma$ ranging from $-1$ to $0.4$.}
	\label{FigureSymmetriczstar2}
\end{figure}

\section{SSVI}\label{ssvi}

	The SSVI is a sub-SVI classically parametrized as
\begin{equation}\label{eqSSVIclassical}
SSVI(k;\theta,\varphi,\rho)= \frac{\theta}{2}\biggl(1+ \rho \varphi k + \sqrt{ (\varphi k+\rho)^2+ 1-\rho^2}\biggr).
\end{equation}
It can be easily checked that the corresponding SVI parameters are
\newline\((a,b,m,\rho,\sigma)=\Bigl(\frac{\theta(1-\rho^2)}{2},\frac{\theta\varphi}{2},-\frac{\rho}{\varphi},\rho,\frac{\sqrt{1-\rho^2}}{\varphi}\Bigr)\),
so that \(\gamma=\sqrt{1-\rho^2}\) and the special property
\(\mu=-\frac{\rho}{\sqrt{1-\rho^2}}=l^*\) hold, where \(l^*\) is the
unique critical point of the smile, which is a point of minimum, of
\(N\). The corresponding SVI formula becomes \begin{equation*}
SVI\biggl(k;b\sigma\sqrt{1-\rho^2},b,\rho,-\frac{\rho\sigma}{\sqrt{1-\rho^2}},\sigma\biggr) = b\sigma\biggl(\sqrt{1-\rho^2}+\rho\Bigl(\frac{k}{\sigma}+\frac{\rho}{\sqrt{1-\rho^2}}\Bigr) + \sqrt{\Bigl(\frac{k}{\sigma}+\frac{\rho}{\sqrt{1-\rho^2}}\Bigr)^2+1}\biggr).
\end{equation*} With our notations,
\(N(l) = \sqrt{1-\rho^2}+\rho l + \sqrt{l^2+1}\).

	We plot in \Cref{FigureSSVI} an SSVI with \(b=1\), \(\rho=\frac{1}{2}\) and
\(\sigma=\frac{1}{2}\).

\begin{figure}
	\centering
	\includegraphics[width=.7\textwidth]{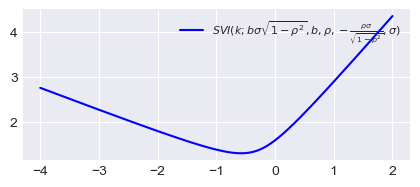}
	\caption{SSVI with $b=1$, $\rho=\frac{1}{2}$ and $\sigma=\frac{1}{2}$.}
	\label{FigureSSVI}
\end{figure}
    
\subsection{The Fukasawa conditions}\label{the-fukasawa-conditions}

	Since \(\gamma>0\), we always have \(\gamma>\tilde F(b,\rho)\). We show
that also the condition on \(\mu\) is verified for \(\gamma>0\).

	\begin{lemma}

Assume $\gamma> 0$. Then for every SVI, $l^*\in I_{\gamma,b,\rho}$.

\end{lemma}

\begin{proof}

Consider the case $\rho\geq 0$, so that $l^* \leq 0$. Note that $N(l)-\gamma = lN'(l) + N''(l)(l^2+1)$. Then the condition $L_-(l)< l^*$ for $l<l^*$ becomes $\frac{\gamma + lN'(l) + N''(l)(l^2+1)}{2N'(l)}(4+bN'(l))< l+l^*$ or equivalently $\gamma>\frac{N'(l)}{4+bN'(l)}(2l^*-l(2+bN'(l))) - N''(l)(l^2+1)$.

The term $- N''(l)(l^2+1)$ is strictly negative and if $2l^*-l(2+bN'(l))\geq0$, also the first term is negative, so that if $\gamma\geq 0$ we automatically have $L_-(l)<l^*$ for every $l<l^*$. When $\rho>0$ (or $\rho=0$ and $b<2$), since $b(1+\rho)\leq2$, it holds $b(1-\rho)\neq 2$. Because $L_-(-\infty)=L_-(l^{*-})=-\infty$ and $L_-$ is continuous, it holds $\sup_{l < l^*}L_-(l) = L_-(l_-)$ where $l_-\in]-\infty,l^*[$. Then if we show $L_-(l)<l^*$ for every $l<l^*$, it follows $l^* > \sup_{l < l^*}L_-(l)$. If $\rho=0$ and $b=2$, it holds $b(1-\rho)=2$, so $l_-=-\infty$ and $L_-(l_-)=-\gamma$. In such case, the only annoying situation arises for $\gamma=0$, because it holds $l^* = 0 = \sup_{l < l^*}L_-(l)$. Taking $\gamma>0$ this case does not arise.

Consider the quantity $2l^*-l(2+bN'(l))$. Since $b(1+\rho)\leq 2$, it is greater than $-2\Bigl(-l^*+l\Bigl(1+\frac{N'(l)}{1+\rho}\Bigr)\Bigr)$, which is positive iff, substituting and dividing by $-2$, the quantity $c(l):=\frac{\rho}{\sqrt{1-\rho^2}}+\frac{l}{1+\rho}\Bigl(1+2\rho+\frac{l}{\sqrt{l^2+1}}\Bigr)$ is negative. Observe that for $l$ going to $-\infty$, $c$ goes to $-\infty$ when $\rho>0$ and to $0^-$ when $\rho=0$. Furthermore $c(l^*)=0$. The derivative of $c$ is $c'(l) = \frac{1}{1+\rho}\Bigl(1+2\rho+\frac{l(l^2+2)}{(l^2+1)^{\frac{3}2}}\Bigr)$. We would like $c'$ to be positive, so that $c$ would be increasing and, from the previous analysis, negative for $l<l^*$. However, this does not always happen. Indeed, the minimum point of $c'$ is reached at $-\sqrt{2}$ and it is $\frac{1}{1+\rho}\Bigl(1+2\rho-\frac{4}3\sqrt{\frac{2}3}\Bigr)$, which is non-negative iff $\rho\geq \Bigl(\frac{2}3\Bigr)^{\frac{3}2}-\frac{1}2=\bar\rho$. Therefore, we are annoyed by the cases where $\rho<\bar\rho$. Under this hypothesis, call $\bar l(\rho)$ the first zero of $c'$, that is also the only zero of $c'$ smaller than $-\sqrt{2}$. We need to prove that $d(\rho):=c(\bar l(\rho);\rho)<0$. Observe that $\frac{d}{d\rho}d(\rho) = \partial_\rho c(\bar l(\rho);\rho) + c'(\bar l(\rho);\rho)\frac{d}{d\rho}\bar l(\rho) = \partial_\rho c(\bar l(\rho);\rho)$. The latter quantity is $\frac{1}{(1-\rho^2)^{\frac{3}2}}+\frac{l}{(1+\rho)^2}\Bigl(1-\frac{l}{\sqrt{l^2+1}}\Bigr)$, which is increasing in $l$ and so it is smaller than its value at $l=-\sqrt{2}$. Then $\partial_\rho c(\bar l(\rho);\rho)\leq \frac{1}{(1-\rho^2)^{\frac{3}2}}-\frac{\sqrt{2}}{(1+\rho)^2}\Bigl(1-\sqrt{\frac{2}{3}}\Bigr)$ and it is negative for $\rho<\bar\rho$. Then, it is enough to show $d(0)<0$. For $\rho=0$, we have $\bar l(0)=-\infty$ and it is easy to see that $d(0)=0^-$ as desired.

Remember that if $l^*>0$ (i.e. $\rho<0$), then $l^*>L_-(l_-)$ because $L_-(l_-)\leq0$ for Lemma 4.4 in \cite{martini2020no}. Then we have proven that $\sup_{l < l^*}L_-(l) < l^*$ for every $\gamma> 0$ and $\rho$.

Using the symmetries, we have that for every $\rho$
$$l^*(\rho)> L_-(l_-(\rho);\rho)=-L_+(-l_-(\rho);-\rho)=-L_+(l_+(-\rho);-\rho)$$
or equivalently $l^*(-\rho)< L_+(l_+(-\rho);-\rho)$.

\end{proof}

	In particular this lemma holds for SSVI and has as immediate consequence
the following:

\begin{lemma}[Fukasawa conditions for SSVI]\label{LemmaFukSSVI}

An SSVI with $0<b\leq \frac{2}{1+|\rho|}$ always satisfies the Fukasawa conditions.

\end{lemma}

\subsection{The condition on $\sigma$}\label{the-condition-on-sigma}

	As in the previous chapters, we start this section by defining some
(already used) useful functions: \begin{align}
b^*(l,\rho) &:= \sqrt{\frac{h(l,\rho)\bigl(h(l,\rho)g_2'(l,\rho) - 2h'(l,\rho)g_2(l,\rho)\bigr)}{g(l,\rho)\bigl(g(l,\rho)g_2'(l,\rho)-2g'(l,\rho)g_2(l,\rho)\bigr)}},\label{eqbstarSSVI}\\
\sigma^*(l,\rho) &:= -\frac{b^*(l,\rho)g_2(l,\rho)}{2(h(l,\rho)^2-b^*(l,\rho)^2g(l,\rho)^2)}.\label{eqsigmastarSSVI}
\end{align}

We will show that the only positive zero of \(g_2\) is
\begin{equation}\label{eql2SSVI}
l_2(\rho) = \frac{1}{\tan\bigl(\frac{\arccos(-\rho)}3\bigr)}.
\end{equation}

The objective of this subsection is to prove the following:

\begin{proposition}[Fully explicit no arbitrage domain for SSVI]\label{PropSSVI}

If an SSVI is arbitrage-free then $0< b\leq \frac{2}{1+|\rho|}$. 
\begin{itemize}
	\renewcommand{\labelitemii}{\tiny$\square$}
	\item If $b=\frac{2}{1+|\rho|}$, an SSVI is arbitrage-free iff $\sigma\geq\sqrt{1-\rho^2}$.
	\item If $b<\frac{2}{1+|\rho|}$, for $\rho\geq 0$, define $\bar l(0,\rho)$ the only root in $]l_2(\rho),\infty[$ of $b^*(l,\rho)=0$.
	
	Then, under the assumption that $\tilde f$ has a unique critical point (that we sustain numerically in \Cref{numerical-proof-of-the-uniqueness-of-the-critical-point-of-tilde-f-in-ssvi}):
	\begin{itemize}
		\item an SSVI with $\rho\geq 0$ is arbitrage-free iff it can be parametrized as
		\begin{equation}\label{eqSVINewParamSSVI}
			SVI(k) = \sigma b^*(l,\rho)\biggl(\sqrt{1-\rho^2}+\rho\Bigl(\frac{k}{\sigma}+\frac{\rho}{\sqrt{1-\rho^2}}\Bigr) + \sqrt{\Bigl(\frac{k}\sigma+\frac{\rho}{\sqrt{1-\rho^2}}\Bigr)^2+1}\biggr)
		\end{equation}
		where $l\in[\bar l(0,\rho),\infty[$ and $\sigma\geq\sigma^*(l,\rho)$;
		\item an SSVI with $\rho<0$ is arbitrage-free iff it can be parametrized as
		\begin{equation}\label{eqSVINewParamSSVI2}
			SVI(k) = \sigma b^*(l,-\rho)\biggl(\sqrt{1-\rho^2}+\rho\Bigl(\frac{k}{\sigma}+\frac{\rho}{\sqrt{1-\rho^2}}\Bigr) + \sqrt{\Bigl(\frac{k}\sigma+\frac{\rho}{\sqrt{1-\rho^2}}\Bigr)^2+1}\biggr)
		\end{equation}
		where $l\in[\bar l(0,-\rho),\infty[$ and $\sigma\geq\sigma^*(l,-\rho)$.
	\end{itemize}
\end{itemize}

\end{proposition}

\subsubsection{Proof of \Cref{PropSSVI}}\label{proof-of}

	We give the proof in the case \(\rho\geq0\). Using
\Cref{Lemmag2rho,Lemmafrho0} we get the following remark:

\begin{remark}\label{RemarkSSVI}

Suppose $\rho\geq0$, then $l_2\leq-l_1$. Fix $l>l_2>0$. For the particular case of SSVI, $h(-l)-h(l)=0$ so the difference between $G_1(l)$ and $G_1(-l)$ is given by $-b^2\frac{\rho l}{4\sqrt{l^2+1}}$, which is negative. Since $g_2(l)\leq g_2(-l)$, then $\tilde f(l)\leq\tilde f(-l)$, so $\inf_{l>l_2}\tilde f(l)\leq\inf_{l< l_1}\tilde f(l)$. It follows then that $\sigma^* = \frac{b}{2\inf_{l>l_2}\tilde f(l)}$.

Similarly for $\rho< 0$ and $l< l_1$, it holds $\tilde f(l)<\tilde f(-l)$ and $\sigma^*=\frac{b}{2\inf_{l< l_1}\tilde f(l)}$.

\end{remark}

\paragraph{Study of $j_2$}\label{study-of-j_2}

	In this section, we work with \(x = \frac{l}{\sqrt{l^2+1}}\), whence
\begin{align*}
\sqrt{1-x^2} N(l(x)) &= 1+ \rho x + \sqrt{1-\rho^2} \sqrt{1-x^2},\\
N'(l(x)) &= x+\rho,\\
N''(l(x)) &= (1-x^2)^{\frac{3}2}.
\end{align*}

\begin{remark}\label{Remarkj2SSVI}

We redefine the useful functions when using the $x$ variable and set: 
$$\Pi(x) := N(l(x)),\ j_2(x) := g_2(l(x)),\ j(x) := g(l(x)),\ J_1(x):=G_1(l(x)),\ \tilde\phi(x) := \tilde f(l(x)).$$
The derivative with respect to $l$ is indicated with $'$, so for example the function $j_2'(x)$ corresponds to $g_2'(l(x))$ rather than $\frac{dj_2}{dx}(x)$.

\end{remark}

	In the study of the function \(g_2\) in \cite{martini2020no}, we recall
that the function has one positive zero \(l_2>l^*\) and it has at least
one point of minimum (that here we prove to be unique and call \(m_2\))
at the right of \(l_2\).

\subparagraph{Computation of $l_2$}\label{computation-of-l_2}

	Let us study the location of the unique positive zero of
\(j_2(x) = (1-x^2)^{\frac{3}2} - \frac{(x+\rho)^2 \sqrt{1-x^2}}{2 (1+\rho x + \bar{\rho} \sqrt{1-x^2})}\),
where \(\bar\rho=\sqrt{1-\rho^2}\). Note that \(j_2(1)=0\); so for
\(x< 1\), \(j_2(x)=0\) iff \begin{equation}\label{zerogx}
2(1-x^2) (1+\rho x + \bar{\rho} \sqrt{1-x^2}) = (x+\rho)^2.
\end{equation} Isolating the radical and squaring yields that the zeros
of \(j_2\) solve the polynomial root equation
\begin{equation}\label{zerogxweak}
4\bar{\rho}^2 (1-x^2)^3 = \bigl( (x+\rho)^2 - 2 (1+\rho x)(1-x^2) \bigr)^2.
\end{equation} The key observation is that \(-\rho\) is a root of
\cref{zerogxweak} and not of \cref{zerogx}. This leads to the following
factorization of the polynomial
\((x+\rho)(4x^5+8\rho x^4+(4\rho^2-3)x^3-5\rho x^2-\rho^2x+\rho^3)\).
Now \(-\rho\) is \emph{again} a root of the rightmost factor which
factors in turn in
\((x+\rho)(4x^4+4\rho x^3-3 x^2 - 2 \rho x + \rho^2)\)\ldots{} and the
miracle continues! Indeed, \(-\rho\) is again a root of this later
factor, leading to the fact that the zeros of \(j_2\) solve
\(\rho = x(3 - 4x^2)\).

	In particular:

\begin{itemize}
\item
  \(\rho=0\), assuming we can exclude \(x=0\), gives
  \(x=\frac{\sqrt{3}}{2}< 1\)
\item
  \(\rho=1\) gives the polynomial equation \(4x^3-3x+1=0\) which reads
  \(4(x+1)\bigl(x-\frac{1}{2}\bigr)^2=0\), whence \(x=\frac{1}{2}\).
\end{itemize}

	Taking the derivative with respect to \(\rho\) gives
\(1=3x'(\rho)(1-4 x(\rho)^2)\) which gives in turn that
\(\rho \to x(\rho)\) is decreasing when \(x>\frac{1}{2}\). We can now
either solve the 3rd degree polynomial root problem, or take \(x\) as a
parameter varying in the range
\(\bigl[\frac{1}{2},\frac{\sqrt{3}}{2}\bigr]\), and backup \(\rho\) by
the formula \(\rho(x) = 3x-4x^3\).

	Going the other route, we note that the polynomial
\(Q_\rho(x):=4x^3-3x+\rho\) satisfies
\(Q_\rho\bigl(\frac{1}{2}\bigr)=-1+\rho<0\),
\(Q_\rho\bigl(\frac{\sqrt{3}}{2}\bigr)=\rho\geq 0\) and also
\(Q_\rho'(x) = 3(4x^2-1)\) which is positive in the range
\(\bigl]\frac{1}{2},\frac{\sqrt{3}}{2}\bigr]\), so it has a unique real
root in this range. Its unique local maximum is located at
\(x=-\frac{1}{2}\) where \(Q_\rho\bigl(-\frac{1}{2}\bigr)=1+\rho>0\),
with, together with the observation that \(Q_\rho(-1)=\rho-1<0\) and
\(Q_\rho(\frac{1}{2})<0\), gives that \(Q_\rho\) has 2 other real roots
located in \(\bigl]-1, -\frac{1}{2}\bigr[\) and
\(\bigl]-\frac{1}{2}, \frac{1}{2}\bigr[\), so that \(x(\rho)\) is the
largest of the roots of \(Q_\rho\). The three solutions have explicit
formula
\[\cos\biggl(\frac{\arccos(-\rho)-2\pi k}3\biggr),\quad k=0,1,2\] and
the greatest is \(x(\rho) = \cos\bigl(\frac{\arccos(-\rho)}3\bigr)\). It
follows that \(l_2(\rho)\) is given by \cref{eql2SSVI}.

\subparagraph{Uniqueness of $m_2$}\label{uniqueness-of-m_2}

	Using the parametrization in \(x\), we get that the zeros of \(j_2'\)
solve: \begin{equation}\label{zeroG2Prime}
(x+\rho)^3 = 2 (1-x^2) \bigl(\rho x +\bar{\rho} \sqrt{1-x^2}+1\bigr) \Bigl( 3x \bigl(\rho x +\bar{\rho} \sqrt{1-x^2}+1\bigr) + x+\rho \Bigr).
\end{equation} In particular for \(\rho=0\) this reads
\(x^2 = 2(1-x^2) \bigl(1+\sqrt{1-x^2}\bigr) \bigl(4+3 \sqrt{1-x^2}\bigr)\).
Isolating the radical and squaring yields the polynomial equation
\(\bigl( x^2+2(1-x^2)(3x^2-7) \bigr)^2 = 4 \times 49 (1-x ^2)^3\); \(0\)
is a root, whereas it is not a solution of \cref{zeroG2Prime}. Letting
\(X=1-x^2\), the polynomial factorizes into \((1-X)^2 (36 X^2-16 X+1)\),
yielding two roots in \(]0,1[\), \(\frac{4+\sqrt{7}}{18}\) and
\(\frac{4-\sqrt{7}}{18}\), with only the latter one solving
\cref{zeroG2Prime}, giving in turn
\(x_{m_2}(\rho=0) = \sqrt{\frac{\sqrt{7}}{18}+\frac{7}{9}}\).

Let us turn now to \(\rho=1\). In this case \cref{zeroG2Prime}
simplifies to \(1=2(1-x)(3x+1)\), or yet \(6x^2-4x-1=0\), yielding
\(x_{m_2}(\rho=1)=\frac{2+\sqrt{10}}{6}\).

Before investigating the general case, we can observe that if we set
\(x=\rho\) in \cref{zeroG2Prime} we get an equation in \(\rho\) which is
\(8 \rho^3 = 4 (1-\rho^2) \times 8 \rho\), or yet
\(\rho^2 = 4 (1-\rho^2)\) which gives
\(x_{m_2}\bigl(\rho=\frac{2}{\sqrt{5}}\bigr)=\frac{2}{\sqrt{5}}\).

	Let us prove now that in the general case \(g_2\) has a unique critical
point.

	Note that \(\partial_\rho N(l) = l+l^*\), so
\(\partial_\rho g_2(l) = -\partial_\rho\frac{N'(l)^2}{2N(l)} = -\frac{N'(l)}{N(l)}\Bigl(1-\frac{N'(l)}{2N(l)}(l+l^*)\Bigr) = -\frac{N'(l)}{N(l)}h(l)< 0\)
because \(h(l)=(G_{1+}(l)+G_{1-}(l))/2>0\). From the formula
\(g_2'(l) = N'''(l)-\frac{N'(l)}{N(l)}g_2(l)\), we find
\(\partial_\rho g_2'(l) = -\Bigl(\partial_\rho\frac{N'(l)}{N(l)}\Bigr)g_2 - \frac{N'(l)}{N(l)}\partial_\rho g_2\)
since \(N'''\) does not depend on \(\rho\). The second term is positive
while the first has the factor
\(\partial_\rho\frac{N'(l)}{N(l)} = \frac{1}{N(l)}\Bigl(1-\frac{N'(l)}{N(l)}(l+l^*)\Bigr) = \frac{1}{N(l)\sqrt{l^2+1}\sqrt{1-\rho^2}} >0\).
Then \(\partial_\rho g_2'(l)>0\).

We have seen that for \(\rho=0\) and \(\rho=1\), the function \(g_2\)
has a unique critical point \(m_2\). Since \(\partial_\rho g_2'(l)>0\),
then for every \(\rho< 1\), the critical points of \(g_2\) must be
greater than \(m_2(\rho=1)=\frac{2\sqrt{2}+\sqrt{5}}{3}\). A critical
point in the \(l\)-variable is still a critical point in the
\(x\)-variable and viceversa, because
\(\frac{df}{dx}(x) = f'(l(x))\frac{dl}{dx}(x) = \frac{f'(l(x))}{(1-x^2)^{\frac{3}{2}}}\)
with the convention \('=\frac{d}{dl}\). Then, showing the convexity of
\(g_2\) in the \(x\)-variable for \(x>x_{m_2}(\rho=1)\) automatically
proves the uniqueness of its critical point in any variable (even if
\(g_2\) is not convex in the \(l\)-variable).

We have
\(\frac{d^2j_2}{dx^2}(x) = \bigl(\frac{dl}{dx}(x)\bigr)^2j_2''(x) + j_2'(x)\frac{d^2l}{dx^2}(x) = (j_2''(x)+3x\sqrt{1-x^2}j_2'(x))(1-x^2)^{-3}\)
and \begin{align*}
j_2''(x)+3x\sqrt{1-x^2}j_2'(x) =& \biggl(\Pi^{iv}(x) - \frac{\Pi''(x)^2}{\Pi(x)} - \frac{\Pi'(x)\Pi'''(x)}{\Pi(x)} + \frac{5\Pi'(x)^2\Pi''(x)}{2\Pi(x)^2} - \frac{\Pi'(x)^4}{\Pi(x)^3}\biggr) +\\&+ 3x\sqrt{1-x^2}\biggl(\Pi'''(x)-\frac{\Pi'(x)\Pi''(x)}{\Pi(x)} + \frac{\Pi'(x)^3}{2\Pi(x)^2}\biggr).
\end{align*} Since \(3x\sqrt{1-x^2} = -\frac{\Pi'''(x)}{\Pi''(x)}\), the
terms \(-\frac{\Pi'(x)\Pi'''(x)}{\Pi(x)}\) and
\(-3x\sqrt{1-x^2}\frac{\Pi'(x)\Pi''(x)}{\Pi(x)}\) simplify. Also,
\(\Pi^{iv}(x) + 3x\sqrt{1-x^2}\Pi'''(x) = 3(1-x^2)^{\frac{5}2}(2x^2-1)\),
which is positive since \(x>\frac{2+\sqrt{10}}{6}>\frac{1}{\sqrt{2}}\).
The term \(- \frac{\Pi''(x)^2}{\Pi(x)}\) becomes positive with the sum
of \(\frac{\Pi'(x)^2\Pi''(x)}{2\Pi(x)^2}\), indeed it becomes
\(-\frac{\Pi''(x)}{\Pi(x)}\Bigl(\Pi''(x) - \frac{\Pi'(x)^2}{2\Pi(x)}\Bigr) = -\frac{\Pi''(x)}{\Pi(x)}j_2(x) > 0\).
Note that the remaining \(4\frac{\Pi'(x)^2\Pi''(x)}{2\Pi(x)^2}\) is
positive. Finally, \(- \frac{\Pi'(x)^4}{\Pi(x)^3}\) compensates with
\(3x\sqrt{1-x^2}\frac{\Pi'(x)^3}{2\Pi(x)^2}\) summing to
\[\frac{\Pi'(x)^3}{\Pi(x)^2}\biggl(\frac{3}2 x\sqrt{1-x^2}-\frac{\Pi'(x)}{\Pi(x)}\biggr) = \frac{\Pi'(x)^3}{2\Pi(x)^3}\bigl(x+3\rho x^2-2\rho+3x\sqrt{1-\rho^2}\sqrt{1-x^2}\bigr).\]
Now \(3\rho x^2-2\rho\) is positive iff \(x>\sqrt{\frac{2}3}\) and this
is true, so also the previous quantity is positive. Then, since the
negative terms of \(\frac{d^2j_2}{dx^2}\) are smaller in magnitude than
its positive terms, \(j_2\) is convex in the \(x\)-variable.

As a consequence, \(j_2\) has a unique critical point for \(x>x_{2}\)
and it lies between \(x_{m_2}(\rho=1)\) and \(x_{m_2}(\rho=0)\).

	In order to study \(x_{m_2}\) in the general case, let us isolate the
radical of \cref{zeroG2Prime} and square. We get the daunting polynomial
root equation \begin{equation*}
\bigl( (\rho+x)^3 - 2 (1-x^2) \bigl( 3 \bar{\rho}^2 x (1-x^2)+(\rho x+1) (3x (\rho x+1)+(\rho+x)) \bigr) \bigr)^2 = 4 \bar{\rho}^2 (1-x^2)^3 \bigl( 6x (\rho x+1) + \rho+x \bigr)^2
\end{equation*} There again, we observe that \(-\rho\) is a root,
whereas it does not solve the initial equation. This leads to an
iterative factorization where eventually \(-\rho\) is a root of order
\(4\). The remaining factor is \begin{equation*}
\rho^2 + 2 \rho (12 x^5-16 x^3 + 5 x) + (36 x^6 - 56 x^4 + 21 x^2).
\end{equation*} For \(x\in[x_{m_2}(\rho=1),x_{m_2}(\rho=0)]\), the only
positive \(\rho\) solution is
\[\rho(x) = x(-12x^4+16x^2-5+2(1-x^2)\sqrt{36x^4-24x^2+1}).\] Since
\(\partial_\rho j_2'<0\), also \(\partial_\rho\frac{dj_2}{dx}<0\) and
the function \(\rho(x)\) is invertible in
\([x_{m_2}(\rho=1),x_{m_2}(\rho=0)] = \Bigl[\frac{2+\sqrt{10}}{6},\sqrt{\frac{\sqrt{7}}{18}+\frac{7}{9}}\Bigr]\)
and its inverse gives the value of \(x_{m_2}\) for fixed \(\rho\).

\paragraph{Study of $J_1$}\label{study-of-j_1}

	From the above remark, \(J_1(x) = \eta^2(x)-b^2j^2(x)\) where
\(\eta(x) = \frac{1}{2}\bigl(1+\sqrt{\frac{1-x^2}{1-\rho^2}}\bigr)\) and
\(j(x) = \frac{x+\rho}{4}\). The simplified formula for \(\eta\) can be
recovered from: \begin{align*}
\eta(x) &= 1-\frac{\Pi'(x)}{2\Pi(x)}\Bigl(\frac{x}{\sqrt{1-x^2}}-\frac{\rho}{\sqrt{1-\rho^2}}\Bigr)\\
&= 1-\frac{(x+\rho)\bigl(x\sqrt{1-\rho^2}-\rho\sqrt{1-x^2}\bigr)}{2\sqrt{1-\rho^2}\bigl(1+\rho x+\sqrt{1-\rho^2}\sqrt{1-x^2}\bigr)}\\
&= \frac{\sqrt{1-x^2}((1-\rho^2)+\rho x+1)+\sqrt{1-\rho^2}((1-x^2)+\rho x+1)}{2\sqrt{1-\rho^2}\bigl(1+\rho x+\sqrt{1-\rho^2}\sqrt{1-x^2}\bigr)}
\end{align*} and collecting at the numerator the quantities
\(\sqrt{1-\rho^2}\sqrt{1-x^2}\) and \(\rho x+1\), one recovers the
product
\(\bigl(\sqrt{1-\rho^2} + \sqrt{1-x^2}\bigr)\bigl(1+\rho x+\sqrt{1-\rho^2}\sqrt{1-x^2}\bigr)\)
where the second factor simplifies with the denominator.

The function \(j\) is of course positive and increasing in \(x\), while
\(\eta\) is positive decreasing in \(x\), since
\(\frac{d\eta}{dx}(x) = -\frac{x}{2\sqrt{1-\rho^2}\sqrt{1-x^2}}<0\).
Then \(J_1\) is a decreasing function and it attains its minimum at
\(1\), equal to
\(\frac{1}{16}\bigl(2-b(1+\rho)\bigr)\bigl(2+b(1+\rho)\bigr)\).

Looking at the condition
\(\sigma\geq\sup_{x>x_2}-\frac{bj_2(x)}{2J_1(x)}\), we then have that
for every \(x>x_2\), it holds
\(-\frac{bj_2(x)}{2J_1(x)} < -\frac{bj_2(m_2)}{2J_1(1)} = -\frac{8bj_2(m_2)}{(4-b^2(1+\rho)^2)}\),
from which one immediately finds a no arbitrage sub-domain for SSVI:

\begin{lemma}[No arbitrage sub-domain for SSVI]\label{lemmaSubDomainSSVI}

Let $j_2$ be given as in \Cref{Remarkj2SSVI}. An SSVI with $0< b(1+|\rho|)< 2$ and $\sigma\geq -\frac{8bj_2(m_2(|\rho|),|\rho|)}{(4-b^2(1+|\rho|)^2)}$ is arbitrage-free.

\end{lemma}

\paragraph{Uniqueness of the critical point of $\tilde f$ for
SSVI}\label{uniqueness-of-the-critical-point-of-tilde-f-for-ssvi}

	Remember that we consider the case \(\rho\geq0\). Then, we look at the
infimum of \(\tilde f(l)\) for \(l>l_2\), letting \(b\) be eventually
\(0\). For \(b<\frac{2}{1+\rho}\), it holds
\(\tilde f(l_2^+) = \infty=\tilde f(\infty)\) and \(\tilde f>0\), so
\(\tilde f\) must have a critical point which is a point of minimum. We
want to prove that in such case \(\tilde f\) has a unique critical point
in \(]l_2,\infty[\). When \(b=\frac{2}{1+\rho}\), it can be shown that
the derivative of \(\tilde f\) vanishes at \(\infty\). We prove that
even in this case this is the unique point of minimum of \(\tilde f\).

	To show that \(\tilde f\) has a unique critical point, we can prove that
\(\frac{d\tilde\phi}{dx}\) has a unique zero.

We have seen that \(J_1\) is decreasing as a function of \(x\). A
critical point must satisfied \(\frac{d\tilde\phi}{dx}=0\) or
\(J_1\frac{dj_2}{dx}-\frac{dJ_1}{dx}j_2=0\). Looking at the sign of the
previous functions, it must hold \(\frac{dj_2}{dx}>0\) so the critical
point is on the right of \(x_{m_2}\). At the critical point, it holds
\(j_2^2\frac{d^2\tilde\phi}{dx^2} = J_1\frac{d^2j_2}{dx^2}-\frac{d^2 J_1}{dx^2}j_2 := n\),
so we should show that \(n\) is positive for every \(x>x_{m_2}\).

We have
\(\partial_{b^2}n = -j^2\frac{d^2j_2}{dx^2} + 2(\frac{dj}{dx}^2+j\frac{d^2j}{dx^2})j_2 = -\frac{\Pi'^2}{16}\frac{d^2j_2}{dx^2} + \frac{j_2}{8} < 0\),
so it is enough to show the positivity of \(n_{|_{b=\frac{2}{1+\rho}}}\)
for \(x>x_{m_2}\) or, more generally, for
\(x>x_{m_2}(\rho=1)=\frac{2+\sqrt{10}}{6}\). We show it numerically in
\Cref{numerical-proof-of-the-uniqueness-of-the-critical-point-of-tilde-f-in-ssvi}. We plot in \Cref{FigureSSVIrhos} the function
\(n_{|_{b=\frac{2}{1+\rho}}}\) for \(x>x_{m_2}(\rho=1)\) for different
values of \(\rho\).

\begin{figure}
	\centering
	\includegraphics[width=.7\textwidth]{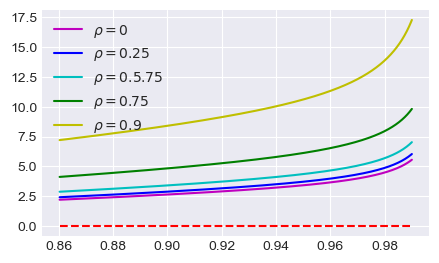}
	\caption{$n'|_{b=\frac{2}{1+\rho}}$ for different $\rho$s.}
	\label{FigureSSVIrhos}
\end{figure}
    
\paragraph{The $b^*$ approach for SSVI}\label{the-b-approach-for-ssvi}

	We have numerically shown that \(\tilde f\) has at most one critical
point in \(]l_2,\infty]\) and this is a point of minimum.

\begin{itemize}
\item
  The limit of \(\tilde f\) at \(\infty\) is finite iff
  \(b=\frac{2}{1+\rho}\) and it equals
  \(\frac{1}{2}\Bigl(\frac{\gamma}{1+\rho}-\mu\Bigr)=\frac{1}{2\sqrt{1-\rho^2}}\)
  for SSVI. So if \(b=\frac{2}{1+\rho}\), the point of minimum of
  \(\tilde f\) is \(\infty\) because there its first derivative is null,
  so that due to uniqueness, the function is strictly decreasing, and in
  such a case, \(\sigma^*=\sqrt{1-\rho^2}\).
\item
  Consider now \(b<\frac{2}{1+\rho}\), \(\tilde f\) has a critical point
  \(\bar l\in]l_2,\infty[\) and it satisfies \(b^2=b^*(\bar l)^2\) or
  \(p(\bar l)=q(\bar l) =0\). In the latter case, we should have
  \(\frac{2h'(\bar l)}{h(\bar l)}=\frac{2g'(\bar l)}{g(\bar l)}\) but
  the right hand side is always positive while we have shown that the
  left hand side is negative. Then we can apply the observation in
  \cref{the-b-approach-reparametrizing-from-the-critical-point-equation} and obtain that the set of the critical points
  of \(\tilde f\) for \(\rho\) fixed is equal to the interval starting
  from the maximum between all the zeros of \(p\), coinciding with
  \(\bar l(0,\rho)\), and ending at
  \(\bar l\bigl(\frac{2}{1+\rho},\rho\bigr) = \infty\). Fix
  \(\rho\geq 0\), \(b<\frac{2}{1+\rho}\), choose
  \(l\in[\bar l(0,\rho),\infty[\). Then
  \(\sigma^*=\frac{b^*(l)}{2\tilde f(l;\rho)}\).
\end{itemize}

\paragraph{Limits of $\bar l$}\label{limits-of-bar-l}

	Because of the uniqueness of the critical points of \(\tilde \phi\),
when \(b\) goes to \(0\), the critical point satisfies
\(\eta j_2'-2\eta'j_2=0\). The numerator of \(\eta j_2'-2\eta'j_2\) is
\begin{multline}\label{eqnump}
\sqrt{1 - \rho^{2}}\bigl(- 4\rho x^{6} + 2(6\rho^{2} - 5)x^5 + 24\rho x^{4} + (31 - 14\rho^{2})x^3 - 13\rho x^{2} + 5(\rho^{2} - 4)x +\rho(\rho^2- 4)\bigr) +\\
+ \sqrt{1 - x^{2}}\bigl(2(2\rho^{2} - 1)x^5 + 4\rho(4 - 3\rho^{2})x^4 + (21 - 22 \rho^{2})x^3 + \rho(8 \rho^{2} - 15)x^2 + 5(3\rho^{2} - 4)x + \rho(3\rho^2 - 4)\bigr).
\end{multline} If \(\rho=0\), possible solutions are \(x=0\) and
\(x=\sqrt{\sqrt{6}-\frac{3}2}\) but since \(x\) must be greater than
\(x_2\), the latter is the searched solution. If \(\rho=1\), the only
possible solution is \(x=1\). In terms of \(l\), these two points
correspond to \(\sqrt{9+4\sqrt{6}}\) and \(\infty\).

	The above expression can also be written as
\[2u\sqrt{1-x^2}\bigl[2x(1-x^2)u-x(x+\rho)^2-\sqrt{1-x^2}(x+\rho+3xu)v\bigr] + (x+\rho)^3v\]
or
\[\varphi(x):=-2u^2\sqrt{1-x^2}\bigl[x\sqrt{1-x^2}\bigl(3\sqrt{1-\rho^2}+\sqrt{1-x^2}\bigr)+x+\rho\bigr] + (x+\rho)^3v\]
where \(u=1+\rho x+\sqrt{1-x^2}\sqrt{1-\rho^2}\) and
\(v=\sqrt{1-x^2}+\sqrt{1-\rho^2}\). Observe that
\((u,v)(x=1)=\bigl(1+\rho,\sqrt{1-\rho^2}\bigr)\) and
\((u,v)(x=\rho)=\bigl(2,2\sqrt{1-\rho^2}\bigr)\) so that

\begin{itemize}
\item
  \(\varphi(1) = (1+\rho)^3\sqrt{1-\rho^2}>0\);
\item
  \(\varphi(\rho) = -48\rho(1-\rho^2)^{\frac{3}2}\leq 0\);
\end{itemize}

and \(\varphi\) has a zero in the range \(]x_2(\rho)\lor \rho, 1[\).

	\begin{lemma}[Computation of $\bar{l}(0,\rho)$]

The only root of $b^*(l) = 0$ is $\sqrt{9+4\sqrt{6}}$ if $\rho=0$, $\infty$ if $\rho=1$ and, in the general case, it lies in $]l_2(\rho)\lor -l^*(\rho),\infty[$. It can be computed as $\bar l(0,\rho) = \frac{\bar x(0,\rho)}{\sqrt{1-\bar x(0,\rho)^2}}$, where $\bar x(0,\rho)$ is the only zero of \cref{eqnump}, which lies in $]x_2(\rho)\lor \rho,1[$.

\end{lemma}

	Remember from \Cref{PrepInverseSVI} that the smile
\(SVI(a,b,-\rho,-m,\sigma)\) is arbitrage-free iff
the smile \(SVI(a,b,\rho,m,\sigma)\) is arbitrage-free. Furthermore, in the case of SSVI, we have the relation \(SVI(k;a,b,-\rho,-m,\sigma) = SVI(-k;a,b,\rho,m,\sigma)\) where
\(a=b\sigma\sqrt{1-\rho^2}\) and
\(m=-\sigma\frac{\rho}{\sqrt{1-\rho^2}}\). So an SSVI with \(\rho<0\) is
arbitrage-free iff the SSVI with parameter \(-\rho>0\) is
arbitrage-free.

\subsubsection{The Gatheral-Jacquier sufficient
conditions}\label{the-gatheral-jacquier-sufficient-conditions}

	How does the boundary in \Cref{lemmaSubDomainSSVI} compare with the
sufficient conditions found by Gatheral and Jacquier in Theorem 4.2 of
\cite{gatheral2014arbitrage}? The theorem asserts that an SSVI is free
of Butterfly arbitrage if \(\theta\varphi(1+|\rho|)<4\) and
\(\theta\varphi^2(1+|\rho|)\leq4\). In terms of the SVI parameters, we
have \(\theta\varphi = 2b\) and \(\varphi=\frac{\sqrt{1-\rho^2}}\sigma\)
so that the two conditions become \begin{align*}
&b(1+|\rho|) < 2,\\
&\sigma \geq\frac{b}{2}(1+|\rho|)\sqrt{1-|\rho|^2},
\end{align*} where the first is the strict Roger-Lee condition. Note
that when \(\sigma\geq\frac{2}{b}\sqrt{\frac{1-|\rho|}{1+|\rho|}}\) and
the Roger-Lee condition is satisfied, then the second condition is
automatically verified and the SSVI is free of arbitrage. Then, we can
simply look at the second condition, keeping in mind that the Roger-Lee
condition \(b(1+|\rho|) \leq 2\) must hold, being a necessary condition.

We can compare with some plots the Gatheral-Jacquier second sufficient
condition with the sufficient condition of \Cref{lemmaSubDomainSSVI} and
the necessary and sufficient condition \(\sigma\geq\sigma^*\). In the
first couple of graphs in \Cref{FigureSSVIGJLemma}, \(\rho\) is fixed while \(b\) ranges in
\(\bigl]0,\frac{2}{1+|\rho|}\bigr[\); in the second couple in \Cref{FigureSSVIGJLemma2}, \(b\) is
fixed and \(\rho\) ranges in \(\bigl]0,1\land\bigl(\frac{2}{b}-1\bigr)\bigr[\).

\begin{figure}
	\centering
	\includegraphics[width=.9\textwidth]{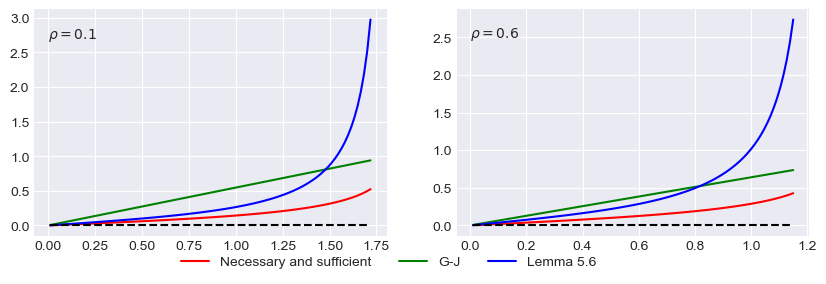}
	\caption{Comparison of Gatheral-Jacquier and \Cref{lemmaSubDomainSSVI} sufficient conditions as functions of $b$.}
	\label{FigureSSVIGJLemma}
\end{figure}

\begin{figure}
	\centering
	\includegraphics[width=.9\textwidth]{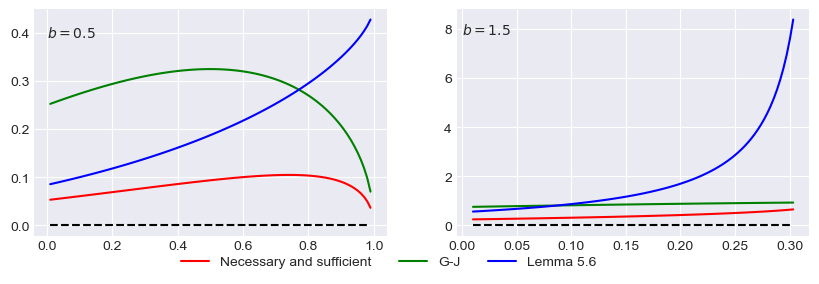}
	\caption{Comparison of Gatheral-Jacquier and \Cref{lemmaSubDomainSSVI} sufficient conditions as functions of $\rho$.}
	\label{FigureSSVIGJLemma2}
\end{figure}
    
	We can see that the sufficient condition of \Cref{lemmaSubDomainSSVI} is
weaker than the Gatheral-Jacquier's one as soon as \(b(1+|\rho|)\)
remains low, but when this quantity goes to \(2\), the former goes to
\(\infty\) since \(J_1(1)\) goes to \(0\).

\subsection{No arbitrage in the Long Term Heston SVI
approximation}\label{no-arbitrage-in-the-long-term-heston-svi-approximation}

In this section, we look at the Long Term Heston SVI introduced by
Gatheral and Jacquier in \cite{gatheral2011convergence}. The authors
prove that, under an appropriate hypothesis on the Heston parameters,
the Heston implied volatility converges to an SVI implied volatility
when the time-to-maturity goes to infinity.

In the Heston model, the underlying process \(S_t\) follows the
stochastic differential equation of the form \begin{align*}
	dS_t &= \sqrt{v_t}S_tdW_t, & S_0>0\\
	dv_t &= \tilde\kappa(\tilde\theta-v_t)dt + \tilde\sigma\sqrt{v_t}dZ_t, & v_0>0\\
	<W,Z>_t &= \tilde\rho dt.
\end{align*} We use the \(\sim\) hat to distinguish Heston parameters
from SVI and SSVI parameters.

Recall the canonical SSVI equation \cref{eqSSVIclassical}. The authors
show that setting \begin{align*}
	\theta &:= \frac{4\tilde\kappa\tilde\theta T}{\tilde\sigma^2(1-\tilde\rho^2)}\Bigl(\sqrt{(2\tilde\kappa-\tilde\rho\tilde\sigma)^2+\tilde\sigma^2(1-\tilde\rho^2)} - (2\tilde\kappa-\tilde\rho\tilde\sigma)\Bigr),\\
	\varphi &:= \frac{\tilde\sigma}{\tilde\kappa\tilde\theta T},
\end{align*} the Heston model converges to an SVI with parameters
\begin{align}\label{LTH}
	a = \frac{\theta}2 (1-\rho^2), & & b=\frac{\theta\varphi}{2}, & & \rho=\tilde\rho, & & m= -\frac{\rho}{\varphi}, & & \sigma=\frac{\sqrt{1-\rho^2}}{\varphi}.
\end{align} Observe that it holds \(\gamma=\sqrt{1-\rho^2}\) and
\(\mu=-\frac{\rho}{\sqrt{1-\rho^2}}=l^*\). In this way it is evident
that the Long Term Heston corresponds to a sub-SVI with three parameters
and, in particular, it corresponds to an SSVI!

We already know from \Cref{LemmaFukSSVI} that an SSVI always satisfies
the Fukasawa conditions. The only requirement needed for a no arbitrage
parametrization is \(\sigma>\sigma^*(b,\rho)\). Then, a Long Term Heston
is free of arbitrage iff it satisfies \cref{PropSSVI}. Since \(\sigma\)
depends linearly to \(T\), indeed
\(\sigma=\frac{\tilde k\tilde\theta\sqrt{1-\rho^2}}{\tilde\sigma}T\), it
increases with \(T\). On the other hand, \(b\) and \(\rho\) are constant
with respect to \(T\) so \(\sigma^*(b,\rho)\) is also constant and
\(\sigma\) will eventually become greater than the arbitrage bound for
\(T\) going to \(\infty\). Then, choosing \(T\) large enough will grant
no arbitrage in the SSVI Long Term Heston parametrization.

This was intuitively true, since the SSVI Long Term Heston approximates
the Heston smile, which is itself arbitrage-free; now the convergence in
the proof of Gatheral and Jacquier is pointwise, so the no arbitrage
property at a fixed \(T\) was still to be proven.

Let us try now to compute a lower bound for \(T\) which grants no
arbitrage.

Suppose \(\rho\geq0\), then from \Cref{RemarkSSVI}, \(f\) reaches its
supremum on the right of \(l_2\), the second zero of \(g_2\). The
requirement \(\sigma\geq -\frac{bg_2(l)}{2(h^2(l)-b^2g^2(l))}\) for
every \(l>l_2\) can be rewritten substituting \(\sigma\) with its
expressions in terms of the Heston parameters, as
\begin{align}\label{eqTHeston}
	T\geq A \sup_{l>l_2}-\frac{g_2(l)}{h^2(l)-b^2g^2(l)}
\end{align}
with
\begin{align*}
	A &:= \frac{b\tilde\sigma}{2\tilde\kappa\tilde\theta\sqrt{1-\rho^2}}=\bigl(\sqrt{(2\tilde\kappa-\rho\tilde\sigma)^2+\tilde\sigma^2(1-\rho^2)} - (2\tilde\kappa-\rho\tilde\sigma)\bigr)/\bigl(\tilde\kappa\tilde\theta(1-\rho^2)^\frac{3}{2}\bigr),\\
	b &= \frac{2}{\tilde\sigma(1-\rho^2)}\Bigl(\sqrt{(2\tilde\kappa-\rho\tilde\sigma)^2+\tilde\sigma^2(1-\rho^2)} - (2\tilde\kappa-\rho\tilde\sigma)\Bigr).
\end{align*}
Note than the function of which we look for the supremum does only
depend on \(\rho\) and on $b$.

The numerator attains its minimum at the locus of the unique minimum of
\(g_2\) on the right of \(l_2\), that we denoted \(m_2(\rho)\). We have
proven that for SSVI, \(h\) is a decreasing function so it is always
greater than \(h(\infty)=\frac{1}{2}\), while \(g\) is increasing with
limit \(g(\infty) = \frac{(1+\rho)}{4}\). In this way, for
\(T\geq -\frac{8b\tilde\sigma g_2(m_2(\rho))}{\tilde\kappa\tilde\theta\sqrt{1-\rho^2}(4-b^2(1+\rho)^2)}\),
the inequality \cref{eqTHeston} is satisfied.

In the case \(\rho<0\), from \Cref{PrepInverseSVI} it follows that the
previous discussion still holds substituting \(\rho\) with \(-\rho\).

We summarize our findings in the following:

\begin{proposition}[No arbitrage sub-domain for the Long Term Heston SVI]\label{propLTH}
	
	The Long Term Heston SVI approximation defined by \cref{LTH} is an SSVI. There is no Butterfly arbitrage in the Long Term Heston SVI approximation as soon as
	$$T \geq -\frac{8b\tilde\sigma g_2(m_2(|\rho|),|\rho|)}{\tilde\kappa\tilde\theta\sqrt{1-\rho^2}(4-b^2(1+|\rho|)^2)}$$
	where $m_2(\rho)$ is the only positive point of minimum of $g_2(\cdot,\rho)$ and
	$$b = \frac{2}{\tilde\sigma(1-\rho^2)}\Bigl(\sqrt{(2\tilde\kappa-\rho\tilde\sigma)^2+\tilde\sigma^2(1-\rho^2)} - (2\tilde\kappa-\rho\tilde\sigma)\Bigr).$$
	
\end{proposition}

Observe that this Proposition corresponds to \Cref{lemmaSubDomainSSVI}
using the Long Term Heston notations.

\newpage
\appendix\section{Numerical proof of the uniqueness of the
critical point of $\tilde f$ in
SSVI}\label{numerical-proof-of-the-uniqueness-of-the-critical-point-of-tilde-f-in-ssvi}

	As discussed in
\cref{uniqueness-of-the-critical-point-of-tilde-f-for-ssvi}, to prove
the uniqueness of the critical point of \(\tilde f\) for SSVI, it is
enough to show that the function
\(n = J_1\frac{d^2j_2}{dx^2}-\frac{d^2 J_1}{dx^2}j_2\) is positive for
every \(x>x_{m_2}(\rho=1)=\frac{2+\sqrt{10}}{6}\) setting
\(b=\frac{2}{1+\rho}\).

The check consists into evaluating the target function at \(\rho\)
spanning from \(0\) to \(0.999\) and \(x\) spanning from
\(\frac{2+\sqrt{10}}{6}\) to \(0.999\), choosing for each variable
\(1000\) points between the extrema. The algorithm is the following:

\begin{lstlisting}[language=Python, basicstyle=\footnotesize]
def N_fun(x, rho): return (1.+rho*x)/np.sqrt(1.-x**2)+np.sqrt(1.-rho**2)

def N1_fun(x, rho): return x+rho

def N2_fun(x): return (1.-x**2)**(3./2.)

def N3_fun(x): return -3.*x*(1-x**2)**2

def N4_fun(x): return 3.*(5.*x**2-1.)*(1.-x**2)**(5./2.)


def h_fun(x, rho): return (1. + np.sqrt((1.-x**2)/(1.-rho**2)))/2.

def hder_fun(x, rho): return -x/(2.*np.sqrt(1.-rho**2)*np.sqrt(1.-x**2))

def g_fun(x, rho): return N1_fun(x,rho)/4.

def gder_fun(): return 1./4.

def hderder_fun(x, rho): return -1./(2.*np.sqrt(1.-rho**2)*(1.-x**2)**(3./2.))
    
def g2_fun(x, rho): return N2_fun(x)-N1_fun(x,rho)**2/(2.*N_fun(x,rho))

def g2derder_fun(x, rho):

    return (x*((rho + x)**3 + 2*(x**2 - 1)*(rho + 3*x*(rho*x + np.sqrt(1 - rho**2)*\\
    np.sqrt(1 - x**2) + 1) + x)*(rho*x + np.sqrt(1 - rho**2)*np.sqrt(1 - x**2) + 1))*\\
    (rho*x + np.sqrt(1 - rho**2)*np.sqrt(1 - x**2) + 1) - 2*np.sqrt(1 - x**2)*\\
    (rho*np.sqrt(1 - x**2) - x*np.sqrt(1 - rho**2))*((rho + x)**3 + 2*(x**2 - 1)*\\
    (rho + 3*x*(rho*x + np.sqrt(1 - rho**2)*np.sqrt(1 - x**2) + 1) + x)*\\
    (rho*x + np.sqrt(1 - rho**2)*np.sqrt(1 - x**2) + 1)) + np.sqrt(1 - x**2)*\\
    (rho*x + np.sqrt(1 - rho**2)*np.sqrt(1 - x**2) + 1)*(np.sqrt(1 - x**2)*(4*x*\\
    (rho + 3*x*(rho*x + np.sqrt(1 - rho**2)*np.sqrt(1 - x**2) + 1) + x)*\\
    (rho*x + np.sqrt(1 - rho**2)*np.sqrt(1 - x**2) + 1) + 3*(rho + x)**2) + 2*(x**2 - 1)*\\
    (rho*np.sqrt(1 - x**2) - x*np.sqrt(1 - rho**2))*(rho + 3*x*(rho*x + np.sqrt(1 - rho**2)*\\
    np.sqrt(1 - x**2) + 1) + x) + 2*(x**2 - 1)*(3*x*(rho*np.sqrt(1 - x**2) -\\
    x*np.sqrt(1 - rho**2)) + np.sqrt(1 - x**2)*(3*rho*x + 3*np.sqrt(1 - rho**2)*\\
    np.sqrt(1 - x**2) + 4))*(rho*x + np.sqrt(1 - rho**2)*np.sqrt(1 - x**2) + 1)))/\\
    (2*np.sqrt(1 - x**2)**3*(rho*x + np.sqrt(1 - rho**2)*np.sqrt(1 - x**2) + 1)**3)


def n_fun(x, b, rho):
    
    h = h_fun(x,rho)
    g = g_fun(x,rho)
    
    return -2*(hder_fun(x,rho)**2+h*hderder_fun(x,rho)-b**2*gder_fun()**2)*g2_fun(x,rho) +\\
    (h**2-b**2*g**2)*g2derder_fun(x,rho)

def n_fun2(x, rho):
    
    return n_fun(x,2./(1.+rho),rho)
    
    
def check_unicity():
    
    rho_check = np.linspace(0.,0.999,1000)
    x_check = np.linspace((2.+np.sqrt(10.))/6.,0.999,1000)

    rho_v, x_v = np.meshgrid(rho_check,x_check)

    x =  np.sum([n_fun2(x_v,rho_v)<0.])
    
    if x == 0.: return 'There is unicity'
    else: return 'No unicity'
\end{lstlisting}

	The result is that at the chosen points, the function is positive,
indeed the command

\begin{lstlisting}[language=Python, basicstyle=\footnotesize]
check_unicity()
\end{lstlisting}

returns:

            \begin{Verbatim}[commandchars=\\\{\}]
 'There is unicity'
\end{Verbatim}


\newpage \bibliography{ImpliedVol}
\bibliographystyle{plain}

\end{document}